\documentclass[lettersize, journal]{IEEEtran}
\usepackage[utf8]{inputenc} 
\usepackage{cite}
\usepackage{amsmath,amssymb,amsfonts}
\usepackage{amsthm}
\usepackage{mathtools}
\allowdisplaybreaks
\usepackage[ruled,vlined]{algorithm2e}
\usepackage{setspace}
\usepackage{caption}
\usepackage{subcaption}
\usepackage{tikz}
\usepackage{tkz-euclide}
\usepackage[long]{optidef}
\usepackage{bm}
\usepackage{fancyhdr}

\newcommand{\primeone}{^{\prime}}
\newcommand{\primetwo}{^{\prime\prime}}
\newtheorem{proposition}{Proposition}
\newtheorem{theorem}{Theorem}
\newtheorem{lemma}{Lemma}
\newtheorem{corollary}{Corollary}
\newtheorem{remark}{Remark}

\DeclareMathOperator{\sinr}{SINR}

\DeclareMathOperator{\sign}{sign}
\newcommand{\C}{\mathbb{C}}

\newcommand{\E}{\mathbb{E}}

\fancypagestyle{firstpage}{%
  \fancyhf{} % clear header/footer
  \fancyhead[C]{%
    \begin{minipage}{\textwidth}
    \centering
    \scriptsize
    This work has been submitted to the IEEE for possible publication. 
    Copyright may be transferred without notice, after which this version may no longer be accessible.
    \end{minipage}%
  }%
}
\begin{document}

\title{Channel-Estimation-Free Beamforming for RIS-Assisted Systems with Limited Feedback}
\author{\IEEEauthorblockN{Hossein Maleki, \IEEEmembership{Graduate Student Member, IEEE} and Hamid Jafarkhani, \IEEEmembership{Fellow, IEEE} \thanks{The authors are with the University of California, Irvine CA, USA. This work was supported in part by the NSF Awards CNS-2229467 and CCF-2328075. Some of the results in this paper were presented at the IEEE International Conference on Communications (ICC-25) \cite{maleki2025lowcomplexity}.}
}
}

\maketitle
\thispagestyle{firstpage}

\begin{abstract}
In this article, we consider received-power maximization in RIS-assisted wireless links using only limited received signal strength (RSS) feedback, avoiding explicit estimation of the high-dimensional cascaded channel. We propose stochastic and deterministic algorithms for passive and active beamforming. In the deterministic method, each RIS element requires only three fixed-phase RSS measurements to recover the phase, followed by tree-structured quantization. We prove a non-asymptotic received-power guarantee for the fixed-beamformer single-user RIS sweep under cascaded Rayleigh fading and extend the bound to finite-resolution RIS phase shifters. Simulations over Rayleigh and temporally correlated channels show that the proposed methods approach full-CSI alternating optimization. We extend the channel-estimation-free framework to multi-user RIS-assisted transmission by proposing a unified single-bit consensus method for both RIS phase adaptation and transmit beamforming. In this mechanism, all users vote on each shared perturbation using only local SINR comparisons, allowing the update to account for both desired-signal improvement and interference generated to other users. The weighted consensus allows for fairness among users based on their priority.

\end{abstract}

\begin{IEEEkeywords}
Reconfigurable intelligent surface (RIS), channel estimation, limited feedback, received signal strength (RSS), RIS phase adaptation, discrete phase shifts, multi-user downlink.
\end{IEEEkeywords}

\section{Introduction}
The next generation of wireless systems pursues ultra-high data rates and reliability through technologies such as massive MIMO \cite{marzetta2010noncooperative, marzetta2016fundamentals}, mmWave communication \cite{rappaport2015millimeter, jiang2018robust, jiang2019mmwave}, and new multiple-access methods \cite{jafarkhani2024mutiple}. Building on intelligent \cite{subrt2012intelligent, hu2017potential} and reconfigurable \cite{cetiner2004multifunctional, cetiner2008method} surfaces, the reconfigurable intelligent surface (RIS) has emerged as a promising enabler \cite{wu2019towards, li2025tutorial}: an array of low-cost passive elements, each programmable to reflect an incident wave with a chosen phase so that the reflections add constructively at the receiver \cite{basar2019wireless, liu2021reconfigurable, wu2021intelligent, elmossallamy2020reconfigurable}.

The main challenge is to acquire channel state information (CSI). An RIS may contain hundreds or even thousands of passive elements, while the transmitter may employ tens to hundreds of antennas, resulting in an enormous number of channels to estimate and feed back. Most existing RIS designs assume perfect CSI \cite{guo2020weighted, yu2019miso}. However, in practice, CSI must first be estimated and then fed back. Moreover, because RIS elements are passive, only the cascaded transmitter-RIS-receiver channel can be observed (at best), whereas the individual transmitter-RIS and RIS-receiver channels cannot be measured separately. These practical limitations have motivated the development of channel-free optimization methods that completely avoid explicit CSI estimation.

CSI-free distributed beamforming with one-bit receiver feedback was originally developed for distributed single-antenna transmitters \cite{mudumbai2009distributed} and later adapted to RIS-aided links with a single-antenna source and destination through random-rotation and one-bit feedback protocols \cite{psomas2021low}. In contrast, we consider a multi-antenna transmitter, where the receiver feedback must drive both active transmit beamforming and passive RIS phase adaptation. Related RIS designs have studied random RIS rotations \cite{zappone2021intelligent}, max-min signal-to-noise ratio (SNR) optimization \cite{kota2024optimal}, dual-function radar-communication operation \cite{yang2026robust}, SNR and outage analyses \cite{cui2021snr,singh2022optimal}, non-orthogonal multiple access (NOMA) rate loss under limited feedback \cite{almasi2023reconfigurable}, robustness to imperfect CSI \cite{zhang2020robust,chen2022robust}, and sum-rate maximization \cite{jiang2021reconfigurable}. Limited-feedback RIS schemes include learning-based \cite{guo2022deep,jiang2021learning,yu2021convolutional}, sparsity-exploiting \cite{shen2021dimension,shi2022triple}, and compressive-sensing-based \cite{shin2022limited} approaches, as well as analyses of channel-estimation error \cite{chen2021adaptive}.
 The closest prior works differ from ours in one of three ways. First, many RIS beamforming methods compute the RIS phases from explicit channel estimates or CSI feedback. Second, robust and quantized-feedback designs study how imperfect or finite-rate CSI degrades performance, but still rely on a channel estimate or a channel-feedback model. Third, CSI-free one-bit feedback methods avoid explicit channel estimation, but they have mainly been developed for single-input single-output (SISO) links or for RIS phase adaptation alone. In contrast, our method jointly updates the multi-antenna transmit beamformer and the RIS phases using only received signal strength (RSS) feedback. In addition, we provide a mathematical convergence analysis. We further extend the proposed design to RIS systems with discrete phase shifts and to multi-user RIS communication, where fairness coefficients are used to prioritize users.
The contributions of this paper are summarized as follows.
\begin{itemize}
\item We derive a full-CSI alternating-optimization (AO) benchmark with closed-form updates and propose RSS-only algorithms: a single-bit transmit beamforming update, a probabilistic single-bit RIS phase update, and a deterministic geometric RIS phase update. The geometric update recovers the per-element phase correction from RSS measurements.

\item We prove a non-asymptotic guarantee for the geometric RIS update in a setting where the beamformer is fixed, only one RIS phase-adaptation sweep is performed, and the probing phases are independent of the channel. The proof establishes a stable reference direction for the aggregate RIS-assisted signal, bounds the phase error of the geometric estimator, and converts this error into a received-power ratio bound. The final bound maintains the element-wise weighting by channel strength, so the analysis reflects which RIS phase errors matter most.

\item We extend the analysis to practical RIS architectures with discrete phase shifters. The result quantifies how finite-resolution phase control degrades the RSS-only update.

\item We extend the RSS-only framework to multi-user RIS communication, where the RIS update must balance the received-power improvements of multiple users rather than optimize a single link. The proposed rule incorporates fairness coefficients to prioritize users and selects each shared RIS update based on the weighted user feedback. We use a weighted consensus vote to update the phases or beamformers according to feedback bits. 

\item All proposed methods bypass explicit channel estimation, which is especially important in RIS-assisted systems, where thousands of channel coefficients may need to be estimated. Instead, the proposed algorithms perform active transmit beamforming and passive RIS phase adaptation directly from receiver feedback.

\end{itemize}

The remainder of the paper is organized as follows. Section~\ref{SysMod} presents the system model, the full-CSI AO benchmark, and a channel-only upper bound. Section~\ref{Lim_Feed} develops the proposed algorithms. Section~\ref{sec:theory} provides the non-asymptotic analysis. Section~\ref{Results} reports simulations, including noise, overhead, and time variation. Section~\ref{Conclusion} concludes the paper, and the appendices contain the proofs.

\section{System Model}
\label{SysMod}
The system has a transmitter with $N_{\rm T}$ antennas and power $P_{\rm T}$, an RIS with $N$ elements, and a single-antenna receiver (Fig.~\ref{fig:system_model}).
\begin{figure}[t]
    \centering
    \includegraphics[width=8cm, page=1, trim=2cm 4cm 2cm 2cm, clip=true]{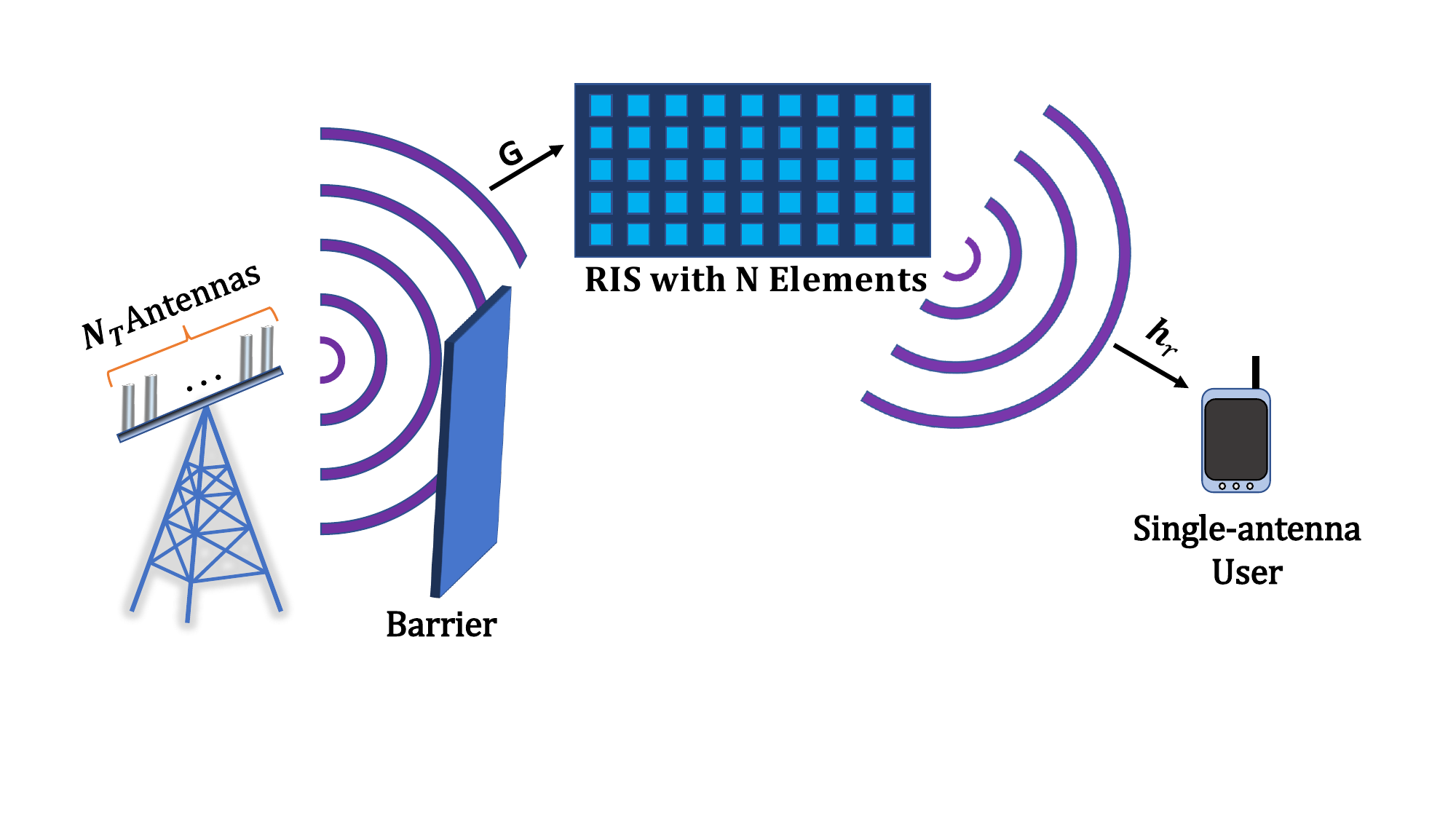}
    \caption{An RIS-assisted system with a multi-antenna transmitter and a single-antenna receiver.}
    \label{fig:system_model}
\end{figure}
The transmitter uses a unit-norm $\boldsymbol{w}\in \mathbb{C}^{N_{\rm T}\times 1}$ to send the unit-energy symbol $s\in \mathbb{C}$. The transmitter-RIS channel is $\boldsymbol{G}\in\mathbb{C}^{N\times N_{\rm T}}$ and the RIS-receiver channel is $\boldsymbol{h}_{\rm r}^{\rm H}\in\mathbb{C}^{1\times N}$, with i.i.d.\ $\mathcal{CN}(0,1)$ entries. The $n$th element shifts its incident phase by $\theta_n$. With the direct link blocked (e.g., due
to high path-loss or heavy shadowing) \cite{kammoun2020asymptotic}, the received signal is
\begin{equation}
    y=\sqrt{P_{\rm T}}\,\boldsymbol{h}_{\rm r}^{\rm H} \boldsymbol{\Theta}\boldsymbol{G}\boldsymbol{w}\,s+n,\quad \boldsymbol{\Theta}=\mathrm{diag}(e^{j\theta_1},\dots,e^{j\theta_N}),
\end{equation}
where $n\sim\mathcal{CN}(0,1)$. The objective is to maximize $\rho=P_{\rm T}|\boldsymbol{h}_{\rm r}^{\rm H}\boldsymbol{\Theta}\boldsymbol{G}\boldsymbol{w}|^2$ subject to $\|\boldsymbol{w}\|=1$ and $|e^{j\theta_n}|=1$. This is a non-convex problem because of the unit-modulus constraints.

\subsection{Alternating-Optimization  Full-CSI Benchmark}
We solve the SNR maximization problem under perfect CSI by alternating optimization and use it as a benchmark. Two subproblems are each solved globally:
\begin{itemize}
\item Beamforming for fixed phases: Defining $\boldsymbol{a}=\boldsymbol{G}^{\rm H}\boldsymbol{\Theta}^{\rm H}\boldsymbol{h}_{\rm r}$, the model is $y=\sqrt{P_{\rm T}}\boldsymbol{a}^{\rm H}\boldsymbol{w}s+n$ and the optimal unit-norm is $\boldsymbol{w}=\boldsymbol{a}/\|\boldsymbol{a}\|$ (maximum-ratio transmission).
\item Phases for fixed beamformer: Defining $\boldsymbol{c}=\boldsymbol{H}_{\rm r}^{\rm H}\boldsymbol{G}\boldsymbol{w}$, with $\boldsymbol{H}_{\rm r} = \mathrm{diag}(\boldsymbol{h}_{\rm r})$ and $c_n=b_ne^{j\beta_n}$, the model is $y = \sqrt{P_{\rm T}}\sum_{n=1}^{N} b_n e^{j(\theta_n+\beta_n)}s + n$, maximized by $\theta_n=-\beta_n$ \cite{jing2009network}.
\end{itemize}
Alternating between these two steps is non-decreasing in SNR and converges to a stationary point. 

\subsection{Cascaded Coefficients and a Channel-Only Upper Bound}
Define the per-element cascaded coefficient (for unit-norm $\boldsymbol{w}$)
\begin{equation}
\label{eqn:zn-def}
z_n \triangleq h_{{\rm r},n}^*\,\boldsymbol{g}_n^{\rm H}\boldsymbol{w},\qquad n=1,\dots,N,
\end{equation}
where $\boldsymbol{g}_n^{\rm H}$ is the $n$th row of $\boldsymbol{G}$, so that $S(\boldsymbol{\theta},\boldsymbol{w})=\sum_n z_n e^{j\theta_n}$ and $\rho=P_{\rm T}|S|^2$. For a fixed $\boldsymbol{w}$, the coherent optimum over phases results in
\begin{equation}
\label{eqn:Sstar}
S^\star(\boldsymbol{w})\triangleq \sum_{n=1}^N |z_n|=\sum_{n=1}^N |h_{{\rm r},n}|\,|\boldsymbol{g}_n^{\rm H}\boldsymbol{w}| .
\end{equation}
To provide a benchmark-independent reference, note that for any unit-norm $\boldsymbol{w}$, $|\boldsymbol{g}_n^{\rm H}\boldsymbol{w}|\le\|\boldsymbol{g}_n\|$ by Cauchy--Schwarz, so the global joint optimum is bounded by the channel-only quantity
\begin{equation}
\label{eqn:UB}
\max_{\|\boldsymbol{w}\|=1}P_{\rm T}\bigl(S^\star(\boldsymbol{w})\bigr)^2 \le U\triangleq P_{\rm T}\Bigl(\sum_{n=1}^N |h_{{\rm r},n}|\|\boldsymbol{g}_n\|\Bigr)^2 .
\end{equation}
Note that the bound $U$ is not necessarily achievable; it is a simple reference to sanity-check, showing that the AO benchmark lies below a valid channel-only upper bound on the true optimum. It can be loose because equality would require a single beamformer to align with every row direction $\boldsymbol{g}_n$ simultaneously. Because the RIS is passive and only the cascaded channel is (noisily) observable, the closed-form full-CSI solution is not directly implementable. The next section develops RSS-only algorithms.

\section{Limited-Feedback Algorithms}
\label{Lim_Feed}
We treat phase adaptation and transmit beamforming separately within an alternating loop. None of the algorithms estimates a channel; all act only on RSS measured at the receiver and fed back over a low-rate link.

\subsection{Probabilistic RIS Phase Adaptation}
\label{sec:prob_phase}
At slot $m=0$ the phases are initialized to an arbitrary $\boldsymbol{\psi}$. In slots $m=1,\dots,M$ a perturbation $\delta_i[m]\sim U[-\Delta,\Delta]$ ($\Delta\in(0,\pi]$) is added for probing. The perturbation is kept if the RSS improves compared to the best RSS so far (single-bit feedback), resulting in
\begin{equation}
    \theta_n[m]=\begin{cases}\theta_n[m-1]+\delta_n[m],&\mathrm{RSS}[m]>\mathrm{RSS}_{\rm best}[m-1],\\\theta_n[m-1],&\text{otherwise.}\end{cases}
\end{equation}
Algorithm~\ref{alg:prob_phase} summarizes the procedure. Since $\mathrm{RSS}_{\rm best}[m]$ is non-decreasing monotone and bounded above by the fixed-$\boldsymbol{w}$ sum $S^\star(\boldsymbol{w})$, the best-RSS sequence always converges. 
\begin{algorithm}
\small
\SetKw{Initialize}{Initialize: }{}{}
\SetKw{Measure}{Measure: }{}{}
\caption{Probabilistic single-bit RIS phase adaptation}\label{alg:prob_phase}
\Initialize{$\boldsymbol{\theta}\leftarrow N\times1$ random, $U[-\pi,\pi]$; $\Delta$; $M$}\\
\Measure{RSS into $\mathrm{RSS}[0]$, $\mathrm{RSS}_{\rm best}[0]$}\\
\For{$m\leftarrow1$ \KwTo $M$}{
$\boldsymbol{\delta}\leftarrow U_{N\times1}[-\Delta,\Delta]$;\\ Transmit with $\boldsymbol{\theta}+\boldsymbol{\delta}$;\\ \Measure{$\mathrm{RSS}[m]$}\\
\eIf{$\mathrm{RSS}[m]>\mathrm{RSS}_{\rm best}[m-1]$}{$\boldsymbol{\theta}\leftarrow\boldsymbol{\theta}+\boldsymbol{\delta}$;\\ $\mathrm{RSS}_{\rm best}[m]\leftarrow\mathrm{RSS}[m]$}{$\mathrm{RSS}_{\rm best}[m]\leftarrow\mathrm{RSS}_{\rm best}[m-1]$}}
\end{algorithm}

\subsection{Deterministic Geometric RIS Phase Adaptation}
\label{sec:TSSQ}
The deterministic algorithm estimates, for each RIS element, the phase offset between that element's reflected signal and the aggregate reflected signal from all other elements. The key geometric fact is that this offset can be recovered from the received powers obtained by applying two known phase perturbations to the element under test.

\begin{proposition}
\label{prop: prop1}
As in Fig.~\ref{fig:vector2}, let $AB$ and $AC$ be non-zero vectors, with lengths $b$ and $a$, respectively. Let $\alpha\in(-\pi,\pi]$ be the angle of $\overrightarrow{AB}$ relative to $\overrightarrow{AC}$. Define $\overrightarrow{AD}=\overrightarrow{AB}+\overrightarrow{AC}$, $\overrightarrow{AD\primeone}=-\overrightarrow{AB}+\overrightarrow{AC}$,
$\overrightarrow{AD\primetwo}=j\overrightarrow{AB}+\overrightarrow{AC}$, corresponding to phase shifts $0$, $\pi$, and $\pi/2$ applied to $AB$. Let $P_0=AD^2$, $P_\pi=AD\primeone{}^2$, $P_{\pi/2}=AD\primetwo{}^2$,
and set $D_1=P_\pi-P_0$, $D_2=P_{\pi/2}-P_0$. Then,
\begin{equation}
    D_1=-4ab\cos\alpha,\qquad
D_2=-2ab(\sin\alpha+\cos\alpha),
\end{equation}
and therefore
\begin{equation}
\label{eqn:atan2}
\alpha=\operatorname{atan2}\bigl(D_1-2D_2,-D_1\bigr).
\end{equation}
Equivalently, when $D_1\neq0$,
\begin{equation}
\label{eqn:prop1}
\tan\alpha=2\frac{D_2}{D_1}-1 .
\end{equation}
Here, $\operatorname{atan2}(y,x)$ denotes the four-quadrant inverse tangent, i.e., the unique angle $\phi\in(-\pi,\pi]$ whose cosine and sine have the signs and ratio specified by the Cartesian coordinates $(x,y)$. Thus, \eqref{eqn:atan2} computes the angle of the vector $(-D_1,\,D_1-2D_2)$. It also remains valid when $D_1=0$ because no division by $D_1$ is required.
\end{proposition}

\begin{proof}
By the law of cosines, $P_0=a^2+b^2+2ab\cos\alpha$,
while the $\pi$-shifted vector is $-\overrightarrow{AB}$, so $P_\pi=a^2+b^2-2ab\cos\alpha$.
Similarly, the $\pi/2$-shifted vector gives $P_{\pi/2}=a^2+b^2-2ab\sin\alpha$.
Thus, $D_1=P_\pi-P_0=-4ab\cos\alpha$, and $D_2=P_{\pi/2}-P_0=-2ab(\sin\alpha+\cos\alpha)$.
Hence, $4ab(\cos\alpha,\sin\alpha)=(-D_1,\,D_1-2D_2)$.
Since $4ab>0$, the angle of this vector is exactly $\alpha$, which proves \eqref{eqn:atan2}. If $D_1\neq0$, dividing the sine component by the cosine component results in
\begin{equation}
    \tan\alpha=\frac{D_1-2D_2}{-D_1} =2\frac{D_2}{D_1}-1,
\end{equation}
which proves \eqref{eqn:prop1}.
\end{proof}

\setlength{\unitlength}{1cm}
\begin{figure}[t]
\centering
\begin{tikzpicture}
  \draw[->, color=blue] (0,0) coordinate (A) node[anchor=east, color=black, yshift=-0.5mm] {$A$} -- (10:4) node[above, color=black, xshift=-0.5mm] {$C$} coordinate (C);
  \draw[->, color=red] (0,0) -- (-30:2) node[right, color=black, xshift=-2mm, yshift=-2mm] {$B$} coordinate (B);
  \draw[->] (0,0) -- (150:2) node[above, color=black, xshift=-1mm, yshift=-1mm] {$B'$} coordinate (B2);
  \draw[->, color=cyan] (0,0) -- (60:2) node[above, color=black] {$B''$} coordinate (B3);
  \draw[->, dashed, color=red] (C) -- +(-30:2) node[right, color=black, xshift=-1mm] {$D$} coordinate (D);
  \draw[dashed] (C) -- +(150:2) node[right, color=black, yshift=1mm] {$D'$} coordinate (D2);
  \draw[->, dashed, color=cyan] (C) -- +(60:2) node[right, color=black, xshift=-1mm] {$D''$} coordinate (D3);
  \draw[dashed, color=red] (B) -- (D);
  \draw[dashed] (B2) -- (D2);
  \draw[dashed, color=cyan] (B3) -- (D3);
  \draw[->, color=red] (0,0) -- (D);
  \draw[->] (0,0) -- (D2);
  \draw[->, color=cyan] (0,0) -- (D3);
  \draw[color=olive, line width=0.5mm] (-30:7mm) node[right=1pt, color=black, yshift=1mm] {$\alpha$} arc (-30:10:7mm);
  \draw[color=violet, line width=0.5mm] (-30:5mm) arc (-30:150:5mm) node[right=1.5pt, midway, above, color=black, xshift=-2mm] {$\beta$};
  \tkzMarkRightAngle[draw=black,size=.3](B3,A,B);
  \tkzLabelAngle[xshift=10mm, yshift=0mm](B3,A,B){$\gamma$};
\end{tikzpicture}
\caption{Geometric construction used to estimate the angle between one RIS element's contribution and the aggregate contribution of the remaining elements.}
\label{fig:vector2}
\end{figure}
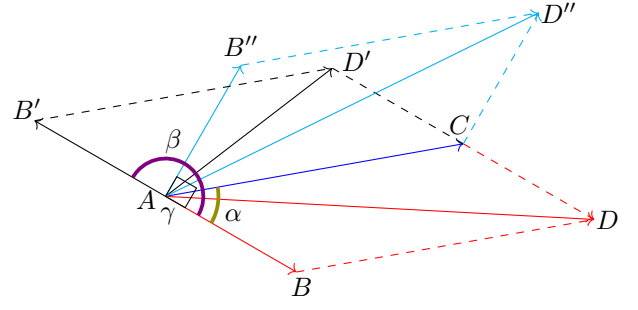

In the RIS algorithm, $AB$ represents the contribution $z_ne^{j\theta_n}$ of Element $n$, and $AC$ represents the leave-one-out aggregate $A_n=\sum_{m\neq n}z_me^{j\theta_m}$. The receiver measures the three powers $P_0$, $P_\pi$, and $P_{\pi/2}$ by keeping all other RIS phases fixed and applying phase shifts $\pi$ and $\pi/2$ only to Element $n$. Then, Eq.~\eqref{eqn:atan2} gives the phase offset $\alpha_n$ between Element $n$ and the aggregate. To align Element $n$ with the aggregate, the RIS updates its phase by $-\alpha_n$, after quantizing it using a tree-structured scalar quantizer (TSSQ)~\cite{gray1998quantization}.

Algorithm~\ref{alg:TSSQ2} uses the proposed method for phase adaptation in the RIS. Note that the leave-one-out reference is frozen during a sweep (no phase is committed until all $\alpha_n$ are computed). The angles are probed once and then refined by $R$ feedback uses. 
\begin{algorithm}
\small
\SetKw{Initialize}{Initialize: }{}{}
\SetKw{Measure}{Measure: }{}{}
\caption{Deterministic geometric + TSSQ phase adaptation}\label{alg:TSSQ2}
\Initialize{$\boldsymbol{\theta}\leftarrow N\times1$ random, $U[-\pi,\pi]$; feedback uses $R$; bits/use $B$}\\
Transmit with $\boldsymbol{\theta}^{\rm TX}\leftarrow\boldsymbol{\theta}$; \Measure{$P_0=\mathrm{RSS}^2$}\\
\For{$n\leftarrow1$ \KwTo $N$}{
$\boldsymbol{\theta}^{\rm TX}\leftarrow\boldsymbol{\theta}$; $\theta^{\rm TX}_n\leftarrow\theta_n+\pi$; Transmit; \\ \Measure{$P_\pi[n]=\mathrm{RSS}^2$}\\
$\boldsymbol{\theta}^{\rm TX}\leftarrow\boldsymbol{\theta}$; $\theta^{\rm TX}_n\leftarrow\theta_n+\pi/2$; Transmit;\\ \Measure{$P_{\pi/2}[n]=\mathrm{RSS}^2$}\\
$D_1[n]\leftarrow P_\pi[n]-P_0$;\\ $D_2[n]\leftarrow P_{\pi/2}[n]-P_0$\\
$\alpha_n\leftarrow\operatorname{atan2}(D_1[n]-2D_2[n],\,-D_1[n])$
}
\For{$k\leftarrow1$ \KwTo $R$}{
\For{$n\leftarrow1$ \KwTo $N$}{Feed back the next $B$ bits of the TSSQ codeword for $\alpha_n$}
}
$\hat{\boldsymbol{\alpha}}\leftarrow$ reconstruction after $L=RB$ total bits/element\\
$\boldsymbol{\theta}\leftarrow\boldsymbol{\theta}-\hat{\boldsymbol{\alpha}}$
\end{algorithm}

\subsection{Tree-Structured Scalar Quantizer}
\label{sec:TSSQ-spec}
Since in general the single scalar angle $\alpha_n\in(-\pi,\pi]$ is uniform \cite{maleki2025lowcomplexity}, we use a uniform scalar quantizer.
%Because the geometric stage reduces the feedback to a single scalar angle $\alpha_n\in(-\pi,\pi]$ per element, 
Such a quantizer is successively refinable \cite{gray1998quantization} and we will have the following properties. %This also clarifies the sense in which we bypass codebook design.
\begin{itemize}
\item \textbf{Source dimension:} The quantized source is a scalar phase $\alpha_n$, not a channel vector.
\item \textbf{Codebook tree:} 
As shown in Fig.~\ref{fig:TSSQ}, the quantizer can be implemented using a tree-structured progressive encoder.
The codebook is a universal binary tree with depth $L$,  obtained by successive bisection of all possible values between $-\pi$ and $\pi$. Every layer of the tree divides its intervals at their midpoints. %The root is $(-\pi,\pi]$ and each node is split at its midpoint. 
A leaf at depth $\ell$ is an interval of width $2\pi/2^{\ell}$ and a  reconstruction at its midpoint. It is fixed a priori, independent of the channel distribution. Practical RIS elements use low-resolution phase shifters that restrict each element to a finite phase alphabet $\Phi_b=\{-\pi+(2m+1)\pi /2^b\}_{m=0}^{2^b-1}$ (spacing $2\pi/2^b$) \cite{tang2023transmissive}. We call the elements of $\Phi_b$ grid points. If we assume a discrete-phase RIS, we can use the same procedure as continuous-phase RIS by using the alphabet values as the encoder reconstruction points. The TSSQ reconstruction levels (its leaves) are chosen to be exactly the phases that the hardware can realize, so its $b$ decisions terminate at a leaf $\widehat\theta_n\in\Phi_b$ %that the element applies directly, 
using $b$ feedback bits per element.
\item \textbf{Progressive feedback:} With $B$ bits per feedback use and $R$ feedback uses, the element receives $L=RB$ bits, and the reconstruction error obeys the bound $|\hat\alpha_n-\alpha_n|\le\pi/2^{L}$ because the depth of the tree is $L$. Note that $B$ may be as small as one bit per use.
\item \textbf{Storage/complexity:} For each RIS element, the encoder and decoder only need to track the current interval endpoints for that element while its feedback bits are being generated or decoded. Thus, the per-element working memory is $O(1)$, the per-bit encoding/decoding cost is $O(1)$. %, and no channel-dependent codebook or lookup table is stored.
\end{itemize}
This method avoids channel-dependent limited-feedback codebook design (the costly part for high-dimensional cascaded channels). 
\begin{figure}[t]
    \centering
    \includegraphics[width=8.5cm, trim=2cm 3cm 2cm 1cm, clip=true]{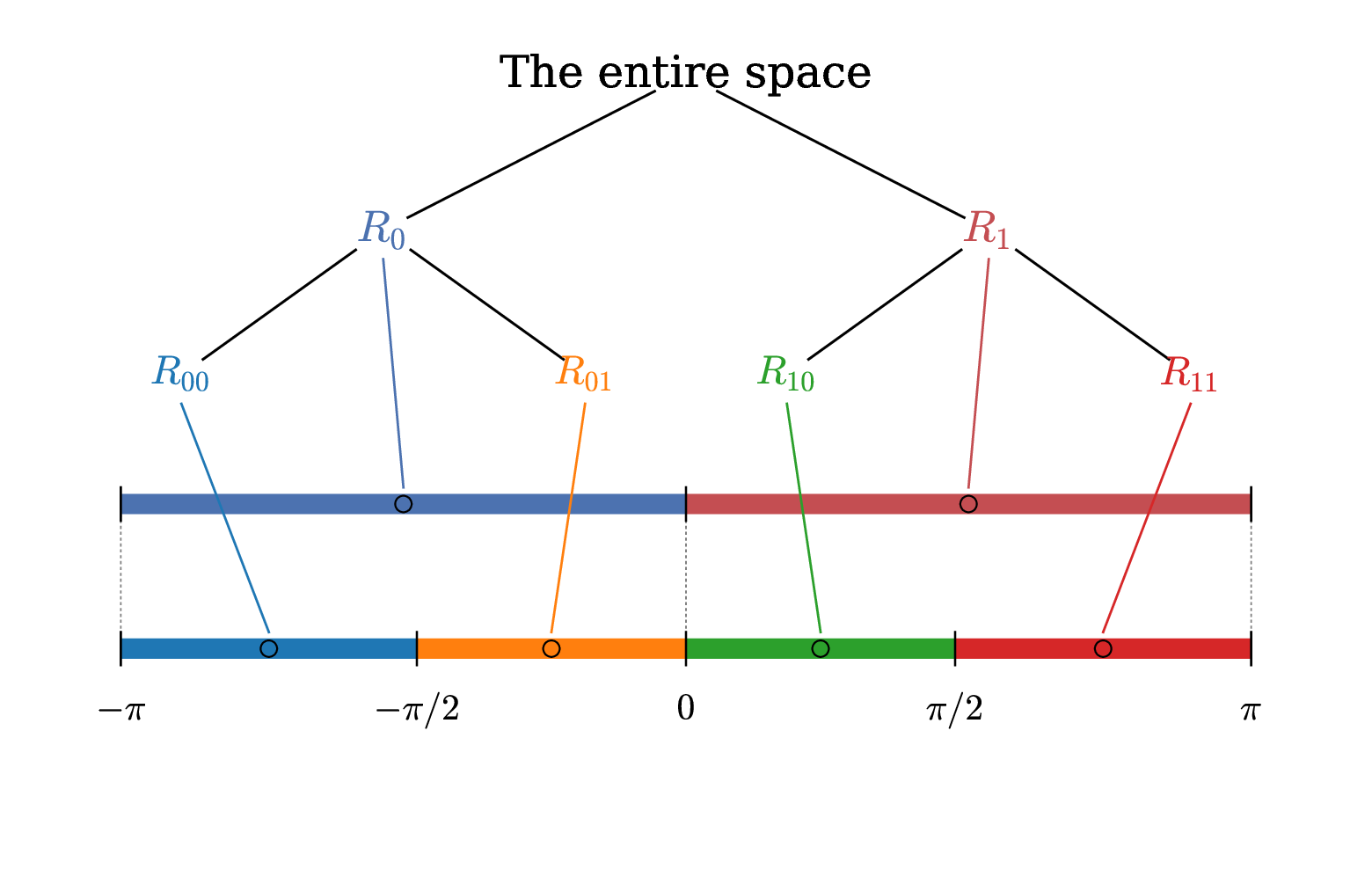}
    \caption{A tree-structured scalar quantizer}
    \label{fig:TSSQ}
\end{figure}

\subsection{Single-Bit Feedback Beamforming}
\label{sec:beamforming}
For fixed phases, the system is the multiple-input single-output (MISO) model $y=\sqrt{P_{\rm T}}\boldsymbol{a}^{\rm H}\boldsymbol{w}s+n$, so we adapt the gradient-sign update of \cite{banister2003simple, spall1992multivariate}. We probe $\boldsymbol{w}\pm\beta\boldsymbol{p}$ with $\boldsymbol{p}\sim\mathcal{CN}(0,2\boldsymbol{I})$ over two slots and keep the unit-norm winner (one feedback bit per iteration). We repeat this procedure for $J$ rounds.
Algorithm~\ref{alg:beamforming} provides the details.
\begin{algorithm}
\small
\SetKw{Initialize}{Initialize: }{}{}
\SetKw{Measure}{Measure: }{}{}
\caption{Single-bit feedback beamforming}\label{alg:beamforming}
\Initialize{$\boldsymbol{w}=\frac{1}{\sqrt{N_{\rm T}}}[1,\dots,1]$; $\beta$; $J$}\\
\For{$j\leftarrow1$ \KwTo $J$}{
$\boldsymbol{p}\leftarrow\mathcal{CN}_{N_{\rm T}\times1}(0,2\boldsymbol{I})$;
$\boldsymbol{w}_e=\tfrac{\boldsymbol{w}+\beta\boldsymbol{p}}{\|\boldsymbol{w}+\beta\boldsymbol{p}\|}$;
$\boldsymbol{w}_o=\tfrac{\boldsymbol{w}-\beta\boldsymbol{p}}{\|\boldsymbol{w}-\beta\boldsymbol{p}\|}$\\
transmit $\boldsymbol{w}_{\rm e},\boldsymbol{w}_{\rm o}$; \eIf{$\mathrm{RSS}_{\rm e}>\mathrm{RSS}_{\rm o}$}{$\boldsymbol{w}=\boldsymbol{w}_{\rm e}$}{$\boldsymbol{w}=\boldsymbol{w}_{\rm o}$}}
\end{algorithm}

\section{Non-Asymptotic Phase-Alignment Analysis}
\label{sec:theory}
This section presents the theoretical performance analysis of the deterministic RIS phase-alignment step. The analysis is carried out for a fixed transmit beamformer and for one RIS measurement/update sweep. This setting isolates the passive phase-alignment mechanism and provides the analytical building block for the alternating design. The fully adaptive joint algorithm, in which both the beamformer and RIS phases are updated from past RSS feedback, is evaluated numerically in Section~\ref{Results}.

\subsection{Setting}
For a fixed, deterministic unit-norm $\boldsymbol{w}$, the coefficients $z_n=h_{\rm{r},n}^*\boldsymbol{g}_n^{\rm H}\boldsymbol{w}$ satisfy $z_n=X_nY_n$ with $X_n,Y_n\stackrel{\mathrm{i.i.d.}}{\sim}\mathcal{CN}(0,1)$, independent across $n$ (since $h_{r,n}\sim\mathcal{CN}(0,1)$ and, for unit-norm $\boldsymbol{w}$, $\boldsymbol{g}_n^{\rm H}\boldsymbol{w}\sim\mathcal{CN}(0,1)$).
The RIS phase vector $\boldsymbol{\theta}$ used during a measurement sweep is deterministic or independent of $\{z_n\}$. This holds at initialization, or for one sweep conditioned on a phase vector chosen independently of the current channel.

The above assumptions define the scope of the analysis. During the RIS sweep analyzed here, the transmit beamformer $\boldsymbol{w}$ is fixed, and the probing phase vector $\boldsymbol{\theta}$ is either deterministic or chosen independently of the current channel realization. Under this condition, the cascaded coefficients $\{z_n\}$ remain independent across RIS elements, which is the probabilistic structure used in the proofs.

%This is different from the full joint adaptive algorithm. In that algorithm, previous RSS feedback is used to update both $\boldsymbol{w}$ and $\boldsymbol{\theta}$. Consequently, at later iterations, these variables are functions of the same channel realization that generates the coefficients $\{z_n\}$, so the independence assumptions used in the one-sweep proof no longer apply directly. Proving convergence in that fully adaptive setting would require a separate analysis conditioned on the feedback history, for example through a martingale argument.
%Therefore, the results in this section are guarantees for one deterministic RIS phase-alignment sweep with fixed $\boldsymbol{w}$. 

First, we need to prove a number of lemmas. Let $S=\sum_m z_m e^{j\theta_m}$ and $A_n=S-z_n e^{j\theta_n}$.
\begin{lemma}
\label{lem:tail}
Let $Z=XY$, where $X,Y\sim\mathcal{CN}(0,1)$ are independent. Then, for every $t>0$,
\begin{equation}
\Pr(|Z|\ge t)\le 2e^{-t}.
\end{equation}
\end{lemma}

\begin{proof}
    See Appendix~\ref{App:A}.
\end{proof}

\begin{lemma}
\label{lem:smallball}
Let $W_M=\sum_{m=1}^M Z_m e^{j\phi_m}$, where $Z_m=X_mY_m$, $X_m,Y_m\stackrel{\mathrm{i.i.d.}}{\sim}\mathcal{CN}(0,1)$ are mutually independent, and the phases $\phi_m$ are deterministic. Then, for every $M\ge2$, $u\in\mathbb{C}$, and $r>0$,
\begin{equation}
\label{eqn:smallball}
\Pr\left(|W_M-u|\le r\right)\le \frac{2r^2}{M}.
\end{equation}
\end{lemma}

\begin{proof}
See Appendix~\ref{App:B}.
\end{proof}

% \begin{lemma}
% \label{lem:quant}
% The depth-$L$ successive-bisection quantizer of Section~\ref{sec:TSSQ-spec} satisfies $|\hat\alpha-\alpha|\le\pi/2^{L}$ for all $\alpha\in(-\pi,\pi]$.
% \end{lemma}
% \begin{proof}
% Let us define the root interval $I_0=(-\pi,\pi]$, of width $|I_0|=2\pi$. At each level, encoding selects the half of the current interval containing $\alpha$, so the level-$\ell$ interval $I_\ell$ satisfies $\alpha\in I_\ell$ and $|I_\ell|=|I_{\ell-1}|/2$; hence, $|I_L|=2\pi/2^L$. The reconstruction $\hat\alpha$ is the midpoint of $I_L$, and since $\alpha\in I_L$, $|\hat\alpha-\alpha|\le|I_L|/2=\pi/2^L$.
% \end{proof}

The next elementary lemma quantifies how much a planar vector's angle can move under an additive perturbation.

\begin{lemma}
\label{lem:angle}
Let $\boldsymbol q\in\mathbb{R}^2\setminus\{0\}$ and
$\boldsymbol e\in\mathbb{R}^2$ with $\|\boldsymbol e\|<\|\boldsymbol q\|$.
Then,
\begin{equation}
\label{eqn:angle-pert}
\left|\angle(\boldsymbol q+\boldsymbol e)-\angle\boldsymbol q\right|
\le
\arcsin\left(\frac{\|\boldsymbol e\|}{\|\boldsymbol q\|}\right)
\le
\frac{\pi}{2}\frac{\|\boldsymbol e\|}{\|\boldsymbol q\|}.
\end{equation}
\end{lemma}

\begin{proof}
The angle difference is invariant under rotation. Therefore, without loss of generality, we assume that $\boldsymbol q=(a,0)$, where $a=\|\boldsymbol q\|$. Let
$\rho=\|\boldsymbol e\|$. Then, $\boldsymbol q+\boldsymbol e$ lies in the disk
of radius $\rho$ centered at $(a,0)$. Since $\rho<a$, this disk does not contain
the origin.

The largest possible angular deviation from the positive real axis occurs along
a tangent line from the origin to this disk. If $\psi$ denotes this tangent
angle, then the radius to the tangent point is perpendicular to the tangent
line and we have
\begin{equation}
\sin\psi=\frac{\rho}{a}.
\end{equation}
Therefore, every point in the disk has an angular deviation at most
$\psi=\arcsin(\rho/a)$, proving the first inequality. The second inequality follows from the convexity of $\arcsin(t)$ on $[0,1]$. Indeed,
$\arcsin(0)=0$ and $\arcsin(1)=\pi/2$, and since a convex function lies below the
chord that connects two points on its graph, we have
\begin{equation}
\arcsin(t)\le (1-t)\arcsin(0)+t\arcsin(1)
=\frac{\pi}{2}t,\quad 0\le t\le1 .
\end{equation}
\end{proof}

\subsection{Finite-$N$ Aggregate Dominance and Direction Stability}
\begin{theorem}
\label{thm:dominance}
Let $N \ge 2$ and define $c_\delta \triangleq \frac{1}{2}\sqrt{\delta}$.
Then, with probability at least $1-\delta$,
\begin{equation}
\label{eqn:dom1}
    \min_n |A_n|
    \ge
    c_\delta \sqrt{N}
    -
    \log\left(\frac{4N}{\delta}\right),
    \qquad
    \max_n |z_n|
    \le
    \log\left(\frac{4N}{\delta}\right).
\end{equation}
If $c_\delta \sqrt{N}\ge 2\log\left(\frac{4N}{\delta}\right)$, then with probability at least $1-\delta$,
\begin{equation}
        \min_n |A_n|
    \ge
    \frac{c_\delta}{2}\sqrt{N},
    \qquad
    \frac{\min_n |A_n|}{\max_n |z_n|}
    \ge
    \frac{c_\delta \sqrt{N}}
         {2\log\left(4N/\delta\right)}.
\end{equation}

\end{theorem}

\begin{proof}
    See Appendix~\ref{App:C}.
\end{proof}
\begin{theorem}
\label{thm:direction}
Under Theorem~\ref{thm:dominance}, if $c_\delta\sqrt N\ge3\log(4N/\delta)$, then with probability at least $1-\delta$,
\begin{equation}
\label{eqn:dir}
\Gamma_N\triangleq\max_n|\angle S-\angle A_n|\le\frac{4}{c_\delta}\frac{\log(4N/\delta)}{\sqrt N}=O\Bigl(\frac{\log N}{\sqrt N}\Bigr).
\end{equation}
\end{theorem}

\begin{proof}
    See Appendix~\ref{App:D}.
\end{proof}

\subsection{Estimator Stability and One-Sweep Bound}
\begin{theorem}
\label{thm:lipschitz}
Let $\alpha^\star = F(d_1^\star,d_2^\star)$, where $F(d_1,d_2)=\operatorname{atan2}(d_1-2d_2,-d_1)$, and define
\begin{equation}
R^\star
\triangleq
\sqrt{(d_1^\star)^2+(d_1^\star-2d_2^\star)^2}
=
4|A_n|\,|z_n|.
\end{equation}
For measured values $\widehat d_i=d_i^\star+\Delta_i$, define $\widehat\alpha=F(\widehat d_1,\widehat d_2)$. If $|\Delta_1|+|\Delta_2|\le R^\star/4$, then
\begin{equation}
\label{eqn:lip}
|\widehat\alpha-\alpha^\star|
\le
\frac{4}{R^\star}\bigl(|\Delta_1|+|\Delta_2|\bigr).
\end{equation}

Now suppose the measurement errors satisfy the tail bound $\Pr(|\Delta_i|>t)\le 2\exp(-t^2/(2\sigma^2))$ for $i=1,2$ and all $t>0$. Equivalently, the errors are sub-Gaussian with variance proxy $\sigma^2$, meaning that their tails are no heavier than those of a zero-mean Gaussian random variable with variance $\sigma^2$. Then, for any $x\in(0,\pi]$,
\begin{equation}
\Pr\left(|\widehat\alpha-\alpha^\star|>x\right)
\le
4\exp\left(-\frac{x^2(R^\star)^2}{128\sigma^2}\right)
+
4\exp\left(-\frac{(R^\star)^2}{128\sigma^2}\right).
\end{equation}
\end{theorem}

\begin{proof}
     See Appendix~\ref{App:E}.
\end{proof}

\begin{theorem}
\label{thm:onesweep}
Fix $\boldsymbol{w}$ and let $\widehat{\boldsymbol\theta}$ be the RIS phases after one alignment sweep. Define the residual phase error of Element $n$ relative to the common coherent direction $\bar\phi$ by
\begin{equation}
\xi_n
\triangleq
\bigl(\angle z_n+\widehat\theta_n\bigr)-\bar\phi
\pmod{2\pi},
\qquad
\xi_n\in(-\pi,\pi].
\end{equation}
Let
\begin{equation}
\eta_n
\triangleq
\varepsilon_n+\Delta_L+\Gamma_N,
\end{equation}
where $\varepsilon_n$ is the geometric-estimation error, $\Delta_L$ is the quantization error, and $\Gamma_N$ is the reference-direction error. If $|\xi_n|\le \eta_n$ for every $n$, then
\begin{equation}
\label{eqn:onesweep0}
|S(\widehat{\boldsymbol\theta},\boldsymbol{w})|
\ge
\sum_{n=1}^N |z_n|\cos\xi_n .
\end{equation}
Moreover,
\begin{equation}
\label{eqn:onesweep}
|S(\widehat{\boldsymbol\theta},\boldsymbol{w})|
\ge
S^\star(\boldsymbol{w})
-
\frac{1}{2}
\sum_{n=1}^N |z_n|\eta_n^2 .
\end{equation}
\end{theorem}

\begin{proof}
By definition of $\xi_n$, the $n$th aligned contribution can be written as
$z_ne^{j\widehat\theta_n}=|z_n|e^{j(\bar\phi+\xi_n)}$. Therefore,
\begin{equation}
S(\widehat{\boldsymbol\theta},\boldsymbol{w})
=
\sum_{n=1}^N z_ne^{j\widehat\theta_n}
=
e^{j\bar\phi}
\sum_{n=1}^N |z_n|e^{j\xi_n}.
\end{equation}
Since multiplication by $e^{j\bar\phi}$ is a global rotation, it does not change the magnitude. Hence,
\begin{equation}
\begin{aligned}
|S(\widehat{\boldsymbol\theta},\boldsymbol{w})|
=
\left|
\sum_{n=1}^N |z_n|e^{j\xi_n}
\right|
&\ge
\operatorname{Re}
\left\{
\sum_{n=1}^N |z_n|e^{j\xi_n}
\right\}
\\
&=
\sum_{n=1}^N |z_n|\cos\xi_n.
\end{aligned}
\end{equation}
Next, using the inequality $\cos t\ge 1-t^2/2$ for all real $t$, we obtain $\cos\xi_n
\ge
1-\frac{\xi_n^2}{2}$.
Since $|\xi_n|\le \eta_n$, we have $\xi_n^2\le \eta_n^2$, and therefore
\begin{equation}
\cos\xi_n
\ge
1-\frac{\eta_n^2}{2}.
\end{equation}
Substituting this bound into \eqref{eqn:onesweep0} gives
\begin{equation}
\begin{aligned}
|S(\widehat{\boldsymbol\theta},\boldsymbol{w})|
&\ge
\sum_{n=1}^N |z_n|
\left(
1-\frac{\eta_n^2}{2}
\right)
\\
&=
\sum_{n=1}^N |z_n|
-
\frac{1}{2}
\sum_{n=1}^N |z_n|\eta_n^2 .
\end{aligned}
\end{equation}
Finally, since $S^\star(\boldsymbol{w})=\sum_{n=1}^N |z_n|$, we obtain
\begin{equation}
|S(\widehat{\boldsymbol\theta},\boldsymbol{w})|
\ge
S^\star(\boldsymbol{w})
-
\frac{1}{2}
\sum_{n=1}^N |z_n|\eta_n^2 ,
\end{equation}
which proves \eqref{eqn:onesweep}.
\end{proof}

\begin{remark}
Theorem~\ref{thm:onesweep} converts the elementwise phase-error bounds into a received-field guarantee after one alignment sweep. The bound holds on the event where the geometric-estimation, quantization, and reference-direction error bounds all hold, so that $|\xi_n|\le\eta_n=\varepsilon_n+\Delta_L+\Gamma_N$ for every element. On this event, the received magnitude is lower bounded by the ideal coherent sum minus a quadratic penalty in the residual phase-error bounds, showing that small phase errors cause only second-order loss in the received field.
\end{remark}

\subsection{Convergence of the Received-Power Ratio}

\begin{theorem}
\label{thm:ratio}
Fix $\boldsymbol{w}$ and $\delta\in(0,1)$. For each element with
$R_n^\star=4|A_n||z_n|>0$, define the measurement-induced phase-error envelope as
\begin{equation}
\label{eqn:eta-meas}
\eta_n^{\mathrm{meas}}
\triangleq
\begin{cases}
\dfrac{4\bigl(|\Delta_{1,n}|+|\Delta_{2,n}|\bigr)}{R_n^\star},
&
|\Delta_{1,n}|+|\Delta_{2,n}|\le R_n^\star/4,
\\[1ex]
\pi,
&
\text{otherwise}.
\end{cases}
\end{equation}
For the degenerate case $R_n^\star=0$, set $\eta_n^{\mathrm{meas}}=\pi$.

Let $\mathcal E_{\mathrm{dir}}$ denote the direction-stability event from
Theorem~\ref{thm:direction}. Under the sufficient condition of
Theorem~\ref{thm:direction}, this event holds with probability at least
$1-\delta$ over the channel. On $\mathcal E_{\mathrm{dir}}$, with
$\Delta_L=\pi/2^L$ and $\Gamma_N$ defined in \eqref{eqn:dir}, the normalized
received magnitude satisfies
\begin{equation}
\label{eqn:weighted-gap}
\frac{|S(\widehat{\boldsymbol\theta},\boldsymbol{w})|}
     {S^\star(\boldsymbol{w})}
\ge
1
-
\frac{1}{2S^\star(\boldsymbol{w})}
\sum_{n=1}^N
|z_n|
\bigl(\eta_n^{\mathrm{meas}}+\Delta_L+\Gamma_N\bigr)^2 .
\end{equation}
Consequently, if
\begin{equation}
\label{eqn:weighted-condition}
\frac{\sum_{n=1}^N |z_n|(\eta_n^{\mathrm{meas}})^2}
     {S^\star(\boldsymbol{w})}
\xrightarrow{\mathbb P}
0,
\qquad
\Delta_L\to0,
\qquad
\Gamma_N\to0,
\end{equation}
then
\begin{equation}
\frac{|S(\widehat{\boldsymbol\theta},\boldsymbol{w})|}
     {S^\star(\boldsymbol{w})}
\xrightarrow{\mathbb P}
1,
\end{equation}
and
\begin{equation}
\frac{\rho(\widehat{\boldsymbol\theta},\boldsymbol{w})}
     {\rho^\star_{\mathrm{RIS}}(\boldsymbol{w})}
\xrightarrow{\mathbb P}
1,
\end{equation}
where
$\rho(\widehat{\boldsymbol\theta},\boldsymbol{w})
\triangleq P_{\rm T}|S(\widehat{\boldsymbol\theta},\boldsymbol{w})|^2$
and
$\rho^\star_{\mathrm{RIS}}(\boldsymbol{w})
\triangleq P_{\rm T}(S^\star(\boldsymbol{w}))^2$.
\end{theorem}

\begin{proof}
The residual phase error $\xi_n$ in Theorem~\ref{thm:onesweep} is bounded by three contributions: the measurement-induced geometric-estimation error $|\widehat\alpha_n-\alpha_n^\star|$, the quantization error, and the reference-direction error. The quantization error is at most $\Delta_L=\pi/2^L$, and by Theorem~\ref{thm:direction}, the reference-direction error is at most $\Gamma_N$ on the direction-stability event $\mathcal E_{\mathrm{dir}}$.

For the geometric-estimation error, Theorem~\ref{thm:lipschitz} gives
$|\widehat\alpha_n-\alpha_n^\star|
\le
4(|\Delta_{1,n}|+|\Delta_{2,n}|)/R_n^\star$
whenever $R_n^\star>0$ and
$|\Delta_{1,n}|+|\Delta_{2,n}|\le R_n^\star/4$. By the definition of
$\eta_n^{\mathrm{meas}}$ in \eqref{eqn:eta-meas}, this implies
$|\widehat\alpha_n-\alpha_n^\star|\le \eta_n^{\mathrm{meas}}$ on this event.
Outside this event, or in the degenerate case $R_n^\star=0$, we set
$\eta_n^{\mathrm{meas}}=\pi$, and the trivial bound
$|\widehat\alpha_n-\alpha_n^\star|\le \pi$ holds because phase errors are measured as principal angle differences in $(-\pi,\pi]$. Hence, in all cases,
$|\widehat\alpha_n-\alpha_n^\star|\le \eta_n^{\mathrm{meas}}$.

Therefore, by the triangle inequality, on $\mathcal E_{\mathrm{dir}}$,
\begin{equation}
|\xi_n|
\le
\eta_n^{\mathrm{meas}}+\Delta_L+\Gamma_N .
\end{equation}
Substituting this elementwise bound into Theorem~\ref{thm:onesweep} and dividing by
$S^\star(\boldsymbol{w})=\sum_{n=1}^N |z_n|$ gives
\begin{equation}
\frac{|S(\widehat{\boldsymbol\theta},\boldsymbol{w})|}
     {S^\star(\boldsymbol{w})}
\ge
1
-
\frac{1}{2S^\star(\boldsymbol{w})}
\sum_{n=1}^N
|z_n|
\bigl(\eta_n^{\mathrm{meas}}+\Delta_L+\Gamma_N\bigr)^2 ,
\end{equation}
which proves \eqref{eqn:weighted-gap}.

It remains to prove the convergence statement. Using
$(a+b+c)^2\le 3(a^2+b^2+c^2)$, we obtain
\begin{equation}
\begin{aligned}
\frac{
\sum_{n=1}^N
|z_n|
\bigl(\eta_n^{\mathrm{meas}}+\Delta_L+\Gamma_N\bigr)^2
}
{2S^\star(\boldsymbol{w})}
&\le
\frac{3}{2}
\frac{
\sum_{n=1}^N |z_n|(\eta_n^{\mathrm{meas}})^2
}
{S^\star(\boldsymbol{w})}
\\
&\quad
+
\frac{3}{2}\Delta_L^2
+
\frac{3}{2}\Gamma_N^2 .
\end{aligned}
\end{equation}
Here, we used
$\sum_{n=1}^N |z_n|\Delta_L^2=\Delta_L^2 S^\star(\boldsymbol{w})$
and
$\sum_{n=1}^N |z_n|\Gamma_N^2=\Gamma_N^2 S^\star(\boldsymbol{w})$.
Under \eqref{eqn:weighted-condition}, together with $\Delta_L\to0$ and
$\Gamma_N\to0$, the right-hand side converges to zero in probability. Therefore,
\eqref{eqn:weighted-gap} implies
\begin{equation}
\frac{|S(\widehat{\boldsymbol\theta},\boldsymbol{w})|}
     {S^\star(\boldsymbol{w})}
\xrightarrow{\mathbb P}
1 .
\end{equation}
Finally, since
$\rho(\widehat{\boldsymbol\theta},\boldsymbol{w})
/
\rho^\star_{\mathrm{RIS}}(\boldsymbol{w})
=
\bigl(|S(\widehat{\boldsymbol\theta},\boldsymbol{w})|/S^\star(\boldsymbol{w})\bigr)^2$,
the continuous-mapping gives
\begin{equation}
\frac{\rho(\widehat{\boldsymbol\theta},\boldsymbol{w})}
     {\rho^\star_{\mathrm{RIS}}(\boldsymbol{w})}
\xrightarrow{\mathbb P}
1 .
\end{equation}
This proves the theorem.
\end{proof}

\begin{remark}
The weighted condition in \eqref{eqn:weighted-condition} does not need to be checked directly through the unknown ideal magnitude $S^\star(\boldsymbol{w})$. Since $
S^\star(\boldsymbol{w})=\sum_{n=1}^N |z_n|$, we have
\begin{equation}
\frac{\sum_{n=1}^N |z_n|(\eta_n^{\mathrm{meas}})^2}
     {S^\star(\boldsymbol{w})}
=
\sum_{n=1}^N
\frac{|z_n|}{S^\star(\boldsymbol{w})}
(\eta_n^{\mathrm{meas}})^2
\le
\max_{1\le n\le N}(\eta_n^{\mathrm{meas}})^2 .
\end{equation}
Therefore, the simpler sufficient condition
\begin{equation}
\max_{1\le n\le N}\eta_n^{\mathrm{meas}}
\xrightarrow{\mathbb P}
0
\end{equation}
implies the weighted condition in \eqref{eqn:weighted-condition}. Thus, the theorem can be interpreted as saying that if the worst elementwise measurement-induced phase error, the quantization error, and the reference-direction error all vanish, then the normalized received magnitude and the corresponding received power both converge to their ideal values.
\end{remark}

\begin{remark}
Theorem~\ref{thm:ratio} explains why the noisy-RSS simulations can remain near-optimal even when some weak elements have poor phase estimates. Their loss is weighted by $|z_n|(\eta_n^{\mathrm{meas}})^2$, so they do not dominate the normalized received magnitude. The result is also consistent with single-bit feedback, since $B=1$ is allowed per feedback use. The quantization error decreases through the total refinement depth $L=KB$, giving $\Delta_L=\pi/2^L$.
\end{remark}

\subsection{Finite-Alphabet (Discrete) Phase Shifters}
\label{sec:discrete}
% Practical RIS elements use low-resolution phase shifters that restrict each element to a finite phase alphabet $\Phi_b=\{-\pi+2\pi m/2^b\}_{m=0}^{2^b-1}$ (spacing $2\pi/2^b$). We call the elements of $\Phi_b$ grid points. This constraint applies throughout the algorithm, i.e., the initial phases, the probe configurations, and every committed update all lie in $\Phi_b$. We therefore build the alphabet into the encoder. The TSSQ reconstruction levels (its leaves) are chosen to be exactly the phases the hardware can realize, so its $b$ decisions terminate at a leaf $\widehat\theta_n\in\Phi_b$ that the element applies directly, using $b$ feedback bits per element and no codebook search.

%Because the phase never leaves the grid, there is only one quantization, and it is immaterial whether one views the update as rounding the ideal aligned phase $\theta_n-\hat\alpha_n$ onto $\Phi_b$ in a single step, or as quantizing the correction and then applying it: with $\theta_n\in\Phi_b$ (and the correction taken on the same grid) the two are identical, since rounding commutes with a shift by a grid point, $Q_b(\theta_n-\hat\alpha_n)=\theta_n-Q_b(\hat\alpha_n)$. 
As explained in Section \ref{sec:TSSQ-spec}, a practical RIS element uses a $b$-bit low-resolution phase shifter, allowing it to select one of $2^b$ discrete phase values. Consequently, the residual phase error of Theorem~\ref{thm:onesweep} carries a single deterministic quantization term $\Delta_b$:
\begin{equation}
\label{eqn:xi-matched}
|\xi_n|\;\le\;\eta_n^{\mathrm{meas}}+\Gamma_N+\Delta_b,
\qquad \Delta_b=\frac{\pi}{2^b}.
\end{equation}

\begin{corollary}
\label{cor:discrete}
Let $b\ge2$ and $\xi_{\max}\triangleq\max_n\bigl(\eta_n^{\mathrm{meas}}
+\Gamma_N+\Delta_b\bigr)$. On the direction-stability event of
Theorem~\ref{thm:ratio}, if $\xi_{\max}\le\pi/2$ then
\begin{equation}
\label{eqn:matched-cos2}
\frac{\widehat\rho_{K,b}}{\rho^\star_{\mathrm{RIS}}(\boldsymbol{w})}
\;\ge\;\cos^2(\xi_{\max}).
\end{equation}
In the noiseless, large-array regime
($\max_n\eta_n^{\mathrm{meas}}\to0$, $\Gamma_N\to0$) this becomes
\begin{equation}
\label{eqn:matched-asymp}
\frac{\widehat\rho_{K,b}}{\rho^\star_{\mathrm{RIS}}(\boldsymbol{w})}
\;\ge\;\cos^2\!\Big(\frac{\pi}{2^{b}}\Big),
\qquad
1-\frac{\widehat\rho_{K,b}}{\rho^\star_{\mathrm{RIS}}(\boldsymbol{w})}
\;\le\;\sin^2\!\Big(\frac{\pi}{2^{b}}\Big).
\end{equation}
\end{corollary}
\begin{proof}
Follows directly from the definitions.
\end{proof}

\section{Multi-User RIS}
In this section, we extend the results to the case with $K$ users. Let us denote the channel between RIS and User $k$ by $\boldsymbol{h}_{\rm r}^{(k)}\in\C^{1\times N}$. The RIS applies $\boldsymbol{\Theta}=\mathrm{diag}(e^{j\theta_1},\dots,e^{j\theta_N})$, and the effective channel between base station (BS) and User $k$ is
\begin{equation}
\label{eq:gk}
\boldsymbol{g}_k(\bm\theta)=\boldsymbol{h}_{\rm r}^{(k)}\boldsymbol{\Theta}\boldsymbol{G}\in\C^{1\times N_{\rm T}}.
\end{equation}
The BS uses a precoder $\boldsymbol{W}=[\boldsymbol{w}_1,\dots,\boldsymbol{w}_K]\in\C^{N_{\rm T}\times K}$ with unit-norm columns, $\|\boldsymbol{w}_k\|=1$. With data symbols $s_j$ ($\E|s_j|^2=1$) and noise $n_k\sim\mathcal{CN}(0,\sigma^2)$, User $k$ receives $y_k=\boldsymbol{g}_k\boldsymbol{w}_ks_k+\sum_{j\neq k}\boldsymbol{g}_k\boldsymbol{w}_js_j+n_k$, and we have
\begin{equation}
\label{eq:sinr}
\sinr_k(\bm\theta,\boldsymbol{W})=\frac{|\boldsymbol{g}_k\boldsymbol{w}_k|^2}{\sigma^2+\sum_{j\neq k}|\boldsymbol{g}_k\boldsymbol{w}_j|^2}.
\end{equation}
Let us define $\phi_k\triangleq\log_2(1+\sinr_k)$. For fairness weights $\alpha_k>0$, the weighted objective is
\begin{equation}
\label{eq:Phi}
\Phi_\alpha(\bm\theta,\boldsymbol{W})\triangleq\sum_{k=1}^K\alpha_k\,\phi_k(\bm\theta,\boldsymbol{W}).
\end{equation}

\subsection{A Unified Single-Bit Consensus}
\label{sec:primitive}

The RIS phases and the transmit beamformers are shared resources. A phase
$\theta_n$ affects every effective user channel in \eqref{eq:gk}, while a
precoding vector $\mathbf w_k$ carries User $k$'s signal and simultaneously
creates interference for all users $j\neq k$. Hence, both passive and active
updates should account for their network-wide effect. We use the following common one-bit consensus rule. A shared variable is probed
in two opposite directions; each user reports whether its own SINR improves, and
the BS updates according to a weighted vote. Let $\mathbf x$ denote the variable being updated, either the RIS phase vector
$\bm\theta$ or the precoder $\mathbf W$. Let $\mathcal R$ be the constraint-enforcing map, let $\mathbf d$ be a probe direction, and let
$\beta>0$ be the probe size. The BS forms $\mathbf x^{+}=\mathcal R(\mathbf x+\beta\mathbf d)$ and 
$\mathbf x^{-}=\mathcal R(\mathbf x-\beta\mathbf d)$.
Each user returns the one-bit comparison
\begin{equation}
\label{eq:bit}
b_k
=
\mathbf 1
\left\{
\sinr_k(\mathbf x^{+})>\sinr_k(\mathbf x^{-})
\right\}.
\end{equation}
Equivalently, since $\phi_k=\log_2(1+\sinr_k)$ is increasing in $\sinr_k$,
this bit also indicates whether User $k$'s rate improves. The BS then forms $S=\sum_{k=1}^{K}\alpha_k(2b_k-1)$,
where $\alpha_k\ge0$ is the fairness or priority weight of User $k$, and
updates
\begin{equation}
    \mathbf x
\leftarrow
\mathcal R\left(\mathbf x+\beta\,\operatorname{sgn}_0(S)\mathbf d\right),
\end{equation}
where $\operatorname{sgn}_0(S)=1$ if $S\ge0$ and
$\operatorname{sgn}_0(S)=-1$ otherwise.
The RIS and beamforming updates are two instances of this method. For the RIS phase update, set $\mathbf x=\bm\theta$, $\mathcal R(\bm\theta)=\bm\theta \bmod 2\pi$, and 
$\mathbf d=\mathbf e_n$.
For a selected RIS element $n$, the BS probes $\theta_n^{+}=\theta_n+\beta_\theta$ and $\theta_n^{-}=\theta_n-\beta_\theta$
with all other phases fixed. Users vote and the BS computes $S_n=\sum_{k=1}^{K}\alpha_k(2b_k-1)$,
and updates $\theta_n
\leftarrow
\theta_n+\beta_\theta\,\operatorname{sgn}_0(S_n)
\quad \bmod 2\pi$.

For the beamforming update, set $\mathbf x=\mathbf W=[\mathbf w_1,\ldots,\mathbf w_K]$ and let $\mathcal R$ normalize the updated beamforming column. To update Column $k$, draw
$\mathbf p_k\sim\mathcal{CN}(\mathbf 0,\mathbf I_{N_{\rm T}})$ and define the matrix perturbation $
\boldsymbol{d} = \mathbf P_k
=
\mathbf p_k\mathbf e_k^T
\in\mathbb C^{N_{\rm T}\times K}$.
Equivalently, $\mathbf P_k$ has $\mathbf p_k$ as its $k$-th column and
zeros in all other columns. The BS forms $\mathbf W^{+}
=
\mathcal R(\mathbf W+\beta_w\mathbf P_k)$ and
$\mathbf W^{-}
=
\mathcal R(\mathbf W-\beta_w\mathbf P_k)$
so that only the $k$-th beamforming column is perturbed. In column notation, $\mathbf w_k^{+}
=
\mathrm{normalize}(\mathbf w_k+\beta_w\mathbf p_k)$ and
$\mathbf w_k^{-}
=
\mathrm{normalize}(\mathbf w_k-\beta_w\mathbf p_k)$,
while all other columns remain fixed. Because
$\mathbf w_k$ not only increases the desired signal of User $k$, but also changes the
interference experienced by users $j\neq k$, all users need to vote. The BS computes $S_k=\sum_{j=1}^{K}\alpha_j(2b_j-1)$
and sets
\begin{equation}
    \mathbf W
\leftarrow
\begin{cases}
\mathbf W^{+}, & S_k\ge0,\\
\mathbf W^{-}, & S_k<0.
\end{cases}
\end{equation}

Updating one beamforming column at a time makes the consensus decision
interpretable: users vote on the network-wide effect of a single shared
degree of freedom. A global perturbation of all columns would instead mix the
effects of desired-signal changes and interference changes across all users,
making one-bit feedback much less informative.

\begin{remark}
\label{rem:selfish}
A selfish column update is obtained by replacing the consensus weights with the
degenerate vector $\bm\alpha=\mathbf e_k$. Then, only User $k$'s vote affects
the update of $\mathbf w_k$. This ignores the interference that
$\mathbf w_k$ causes to the other users. The consensus rule retains the same
one-bit feedback structure, but allows users harmed by the perturbation to vote
against it.
\end{remark}

\subsection{One-Bit Vote Consistency} In this subsection, we want to show that a single sign bit, taken with a small symmetric probe, reveals the sign of a directional derivative.

\begin{lemma}
\label{lem:vote}
Fix $(\bm\theta,\mathbf{W})$, an element $n$, and a user $k$, and let $g\triangleq\partial\sinr_k/\partial\theta_n$. The difference $D_k(\beta)\triangleq\sinr_k(\bm\theta+\beta\mathbf e_n,\mathbf W)-\sinr_k(\bm\theta-\beta\mathbf e_n,\mathbf W)$ satisfies
\begin{equation}
\label{eq:Dexp}
D_k(\beta)=2\beta\,g+O(\beta^3)\qquad(\beta\to0).
\end{equation}
Consequently, if $g\neq0$, there is $\beta_0>0$ such that for all $0<\beta_\theta<\beta_0$, the feed-back bit obeys
\begin{equation}
    2b_k-1=\sign(g)=\sign\Bigl(\frac{\partial\phi_k}{\partial\theta_n}\Bigr).
\end{equation}
\end{lemma}
\begin{proof}
Let $f(\beta)=\sinr_k(\bm\theta+\beta\mathbf e_n,\mathbf W)$. Then, $D_k(\beta)=f(\beta)-f(-\beta)$ is an odd function of $\beta$, so its Maclaurin expansion contains only odd powers: 
\begin{equation}
    D_k(\beta)=2f'(0)\beta+\tfrac{2}{3!}f'''(0)\beta^3+\cdots=2\beta g+O(\beta^3),
\end{equation}
 giving \eqref{eq:Dexp}. If $g\neq0$, then $D_k(\beta)=2\beta g\,(1+O(\beta^2))$, so there is $\beta_0>0$ with $\sign(D_k(\beta))=\sign(g)$ for all $0<\beta<\beta_0$. The bit \eqref{eq:bit} is $b_k=\mathbf 1\{D_k>0\}$, hence $2b_k-1=\sign(g)$. Finally, $\phi_k=\log_2(1+\sinr_k)$ is strictly increasing in $\sinr_k$, so $\partial\phi_k/\partial\theta_n=g/((1+\sinr_k)\ln2)$ has the same sign as $g$.
\end{proof}

\begin{lemma}
\label{lem:spsa}
Equip $\mathbb C^{N_{\rm T}\times K}$ with the real inner product
\begin{equation}
\label{eq:realip}
\langle \mathbf A,\mathbf B\rangle_{\mathbb R}
\triangleq
\Re\operatorname{tr}(\mathbf A^{\mathsf H}\mathbf B).
\end{equation}
Let
\[
\nabla_{\mathbf W}\sinr_k
\triangleq
2\frac{\partial \sinr_k}{\partial \mathbf W^*}
\]
denote the Wirtinger conjugate gradient, so that the directional derivative of
$\sinr_k$ along a complex perturbation $\mathbf P$ is
\begin{equation}
\label{eq:dirderiv}
\left.
\frac{d}{d\beta}
\right|_{\beta=0}
\sinr_k(\bm\theta,\mathbf W+\beta\mathbf P)
=
\left\langle
\nabla_{\mathbf W}\sinr_k,\mathbf P
\right\rangle_{\mathbb R}.
\end{equation}
Fix $(\bm\theta,\mathbf W)$ and a perturbation direction $\mathbf P$. Define
the SINR difference
\begin{equation}
    D_k^W(\beta)
\triangleq
\sinr_k(\bm\theta,\mathbf W+\beta\mathbf P)
-
\sinr_k(\bm\theta,\mathbf W-\beta\mathbf P).
\end{equation}
If $\sinr_k$ is three times continuously differentiable in a neighborhood of
$\mathbf W$, then
\begin{equation}
\label{eq:spsa}
D_k^W(\beta)
=
2\beta
\left\langle
\nabla_{\mathbf W}\sinr_k,\mathbf P
\right\rangle_{\mathbb R}
+
O(\beta^3).
\end{equation}
Consequently, whenever $\beta$ is sufficiently small and
$
\left\langle
\nabla_{\mathbf W}\sinr_k,\mathbf P
\right\rangle_{\mathbb R}
\neq 0$,
the comparison bit $\mathbf 1\{D_k^W(\beta)>0\}$
reveals the sign of the directional derivative of $\sinr_k$ along
$\mathbf P$.
\end{lemma}

\begin{proof}
Let $f(\beta)
=
\sinr_k(\bm\theta,\mathbf W+\beta\mathbf P)$.
By the smoothness assumption, $f$ admits a Taylor expansion around $\beta=0$. Its first derivative is
$
f'(0)
=
\left\langle
\nabla_{\mathbf W}\sinr_k,\mathbf P
\right\rangle_{\mathbb R},
$
by \eqref{eq:dirderiv}. Following the proof of Lemma~\ref{lem:vote},
\begin{equation}
\label{eq:two_sided_bf_difference}
D_k^W(\beta)
=
f(\beta)-f(-\beta)
=
2\beta f'(0)+O(\beta^3).
\end{equation}
Substituting the expression for $f'(0)$ yields \eqref{eq:spsa}. If
$f'(0)\neq 0$, then for sufficiently small $\beta$, the linear term
dominates the remainder $O(\beta^3)$, and
\[
\sign D_k^W(\beta)
=
\sign f'(0)
=
\sign
\left\langle
\nabla_{\mathbf W}\sinr_k,\mathbf P
\right\rangle_{\mathbb R}.
\]
\end{proof}

The previous two lemmas provide the local meaning of the one-bit feedback. For a sufficiently small symmetric probe, the bit returned by User $k$ indicates whether that user prefers the positive or negative direction, to the first order. The feedback reveals the sign of each user's local change, but not its magnitude. Therefore, in the general interference-limited multi-user setting, the proposed rule can be interpreted as a coarse one-bit approximation to the ascent direction of the weighted performance objective $\Phi_\alpha$. Since each user reports only a binary comparison, the BS does not observe the exact directional derivatives $\partial_{\mathbf d}\phi_k$ or their magnitudes. Instead, it observes their signs through the local probe. The weighted-majority score therefore acts as a sign-only surrogate for the true weighted directional derivative
\begin{equation}
\label{eq:true_weighted_directional_derivative}
\partial_{\mathbf d}\Phi_\alpha
=
\sum_k \alpha_k \partial_{\mathbf d}\phi_k .
\end{equation}
When the users' local directional preferences are broadly aligned, or when the majority direction also captures the dominant weighted derivative contribution, this surrogate selects the same first-order direction as the weighted-sum objective. Thus, the consensus step provides a low-overhead approximation of gradient-based coordination using only single-bit user feedback.

In the single-user case, this approximation becomes exact at the sign level. When $K=1$, the weighted-majority score reduces to the single user's bit, and Lemmas~\ref{lem:vote} and~\ref{lem:spsa} show that the update follows the local ascent direction for sufficiently small probes. Hence, the proposed primitive reduces to the standard single-user single-bit adaptation rule for $K=1$, while for multiple users, it provides a sign-based scalable approximation to the coordinated weighted-gradient adaptation.

\section{Simulation Results}
\label{Results}
For simulations, the channels are i.i.d. Rayleigh, i.e., the entries of $\boldsymbol{G}$ and $\boldsymbol{h}_{\rm r}$, are i.i.d. $\mathcal{CN}(0,1)$. The initial RIS phases are $U(-\pi,\pi]$, $\boldsymbol{w}=\frac{1}{\sqrt{N_{\rm T}}}[1,\dots,1]$, TSSQ uses $B=1$ bit/use and $R=6$ (so $L=6$ bits/element, i.e., 6-level TSSQ), and the beamforming step is $\beta=0.01$ ($N_{\rm T}=4$) or $0.003$ ($N_{\rm T}=16$). Throughout this section, "benchmark" denotes the AO full-CSI method of Section~\ref{SysMod}.

Fig.~\ref{fig:convergence} shows convergence results for a choice of transmit antennas and RIS elements for one user. To place our method in context, we compare it with two families of channel-quantization baselines. The first baseline applies a scalar quantizer (SQ), quantizing each channel coefficient independently. We consider resolutions of $1$, $2$, $3$, and $4$ bits per number, with quantizer levels optimized through the Lloyd algorithm~\cite{gray1998quantization}. Since the system has $(N_{\rm T}+1)N$ complex channels, equivalently $2(N_{\rm T}+1)N$ real numbers, the feedback cost grows quickly. In the setting of Fig.~\ref{fig:convergence}, for example, a $2$-bit SQ must encode $10000$ numbers, amounting to $20000$ bits. Our TSSQ-based algorithm reaches comparable performance in only $3842$ time slots, which corresponds to fewer than $3842$ feedback bits, as the beamforming stage transmits one bit every two slots. The second baseline is a random vector quantizer (RVQ)~\cite{gray1998quantization}. RVQ outperforms SQ when the vectors are low-dimensional, but it becomes difficult to apply as the dimension grows. We therefore quantize the rows of $\mathbf{G}$ directly and split the channel vector $\mathbf{h}_{\rm r}$ into subvectors of length $4$. For each vector, $B_1$ bits encode its direction and $B_2$ bits its magnitude. Because the magnitude is relatively easy to represent, it suffices to use $B_2=2$ bits. As shown in Fig.~\ref{fig:convergence}, our approach outperforms both the SQ and RVQ baselines. Also, it is important to note that we assumed perfect channel estimation for SQ and RVQ; otherwise, their performance would degrade more.

\begin{figure}[h]
    \centering
    \includegraphics[width=0.95\linewidth]{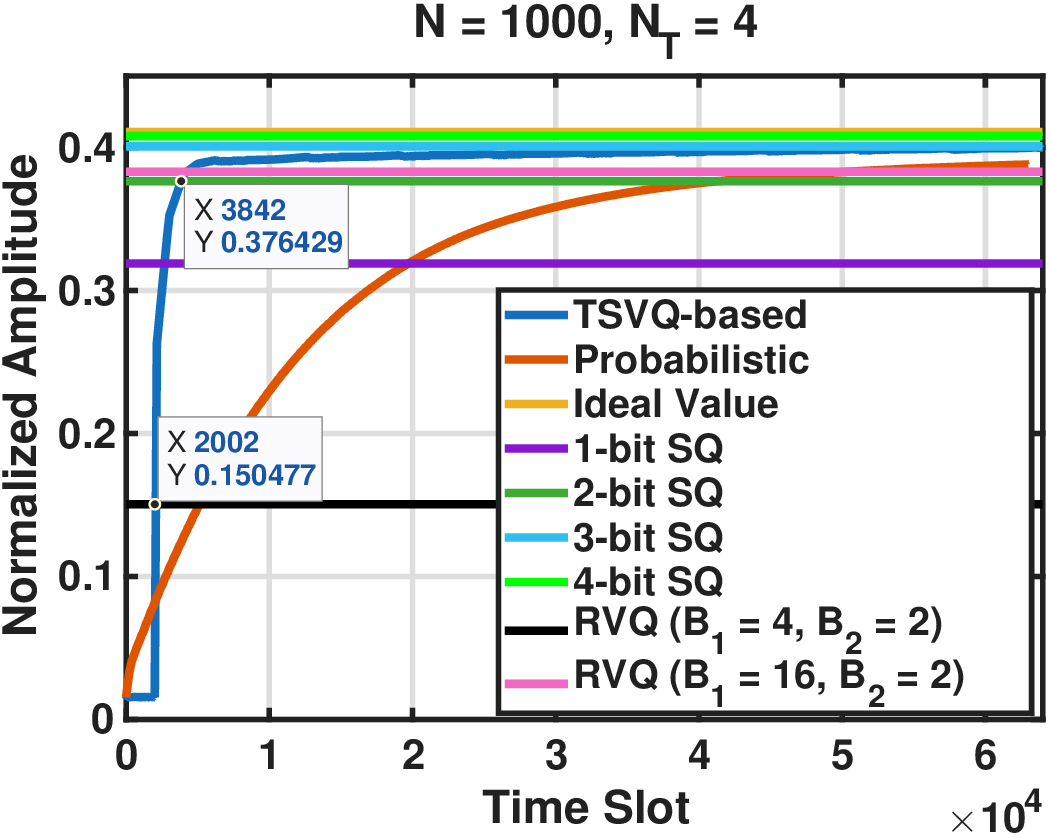}
    \caption{Convergence of the limited feedback algorithm to the optimal value: Results for $100$ iterations of Monte Carlo simulation with $1000$ RIS elements, $4$ transmit antennas, and one user.}
         
    \label{fig:convergence}
\end{figure}

Time variation adds a further layer of difficulty to RIS-assisted systems, since achieving the optimal operating point requires both the transmitter and the RIS to track the channel at every instant. Fig.~\ref{fig:time_varying} illustrates a hybrid scheme designed to cope with this setting. During the first alternation of the RIS phase adaptation, we run our TSSQ-based algorithm to achieve a good initial configuration. Afterward, to follow the channel as it drifts, we switch to a probabilistic update~\cite{mudumbai2009distributed}. At each step, a random perturbation is applied to the phases, and a single feedback bit reports whether the received signal improved or worsened. For comparison, we also show the case in which only probabilistic phase perturbations are used throughout. The channel evolves according to a first-order autoregressive (AR(1)) model, $G[nN_S] = a_1\, G[(n-1)N_S] + b_1 u_1$, and $h_{\rm r}[nN_S] = a_2\, h_{\rm r}[(n-1)N_S] + b_2 u_2$, 
where $a_1, a_2, b_1, b_2$ are constants and $u_1, u_2$ are independent
innovations with i.i.d.\ $\mathcal{CN}(0,1)$ entries of the appropriate
dimensions. With these Gaussian innovations, the model is equivalently known as a first-order Gauss-Markov model. Here, $N_S$ denotes the interval over which the channel remains static. We set
$a_1 = a_2 = 0.9/\sqrt{0.9^2 + 0.1^2}$ and
$b_1 = b_2 = 0.1/\sqrt{0.9^2 + 0.1^2}$. Following 3GPP TS~38.211 (v18.9.0, Release~18), an orthogonal frequency division multiplexing (OFDM) frame with $15$~kHz
subcarrier spacing carries $140$ symbols; therefore, we treat the channel as static within a frame and independently refreshed from one frame to the next. Using this model, Fig.~\ref{fig:time_varying_single} reports the tracking performance of the hybrid algorithm for a single-channel realization. Across $100$ Monte Carlo runs, our algorithm
consistently attains roughly $80$ to $90$ percent of the optimal performance (Fig.~\ref{fig:time_varying_monte}).

\begin{figure}[t]
    \centering
    \begin{subfigure}{0.95\linewidth}
        \centering
        \includegraphics[width=\linewidth]{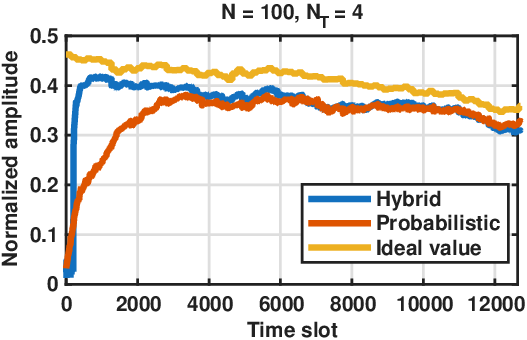}
        \caption{Single realization}
        \label{fig:time_varying_single}
    \end{subfigure}

    \vspace{1ex}

    \begin{subfigure}{0.95\linewidth}
        \centering
        \includegraphics[width=\linewidth]{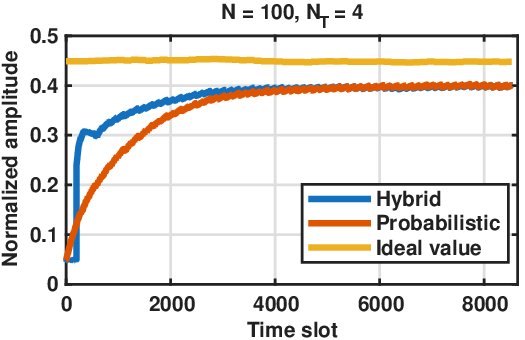}
        \caption{Monte Carlo average (100 realizations)}
        \label{fig:time_varying_monte}
    \end{subfigure}
    \caption{Performance of limited-feedback algorithms in a standard time-varying channel: (a) a single channel realization and (b) the Monte Carlo average over 100 realizations.}
    \label{fig:time_varying}
\end{figure}

Fig.~\ref{fig:dominance} examines the finite-$N$ behavior of the geometric estimator. Our analysis predicts that the ratio $\min_n \frac{|A_n|}{|z_n|}$ grows on the order of $\sqrt{N}/\log N$, so that element-wise dominance improves as the array becomes larger.

To verify this prediction, we sweep the number of RIS elements $N$ over a logarithmic grid and, for each $N$, average over $200$ independent channel realizations. The horizontal axis reports $N$ and the vertical axis reports the dominance ratio, both on a logarithmic scale. We plot two empirical curves, the mean ratio and its $10$th percentile, the latter capturing the unfavorable tail of the distribution. For reference, we also draw the theoretical lower bound
$(c_\delta/2)\sqrt{N}/\log(4N/\delta)$, with $c_\delta = \sqrt{\delta}/2$ and $\delta = 0.1$. As Fig.~\ref{fig:dominance} shows, both empirical curves remain above the bound and follow the same $\sqrt{N}/\log N$ trend across nearly three orders of magnitude in $N$, confirming that aggregate dominance strengthens with the array size exactly as predicted by the theory.

\begin{figure}[h]
    \centering
    \includegraphics[width=0.95\linewidth]{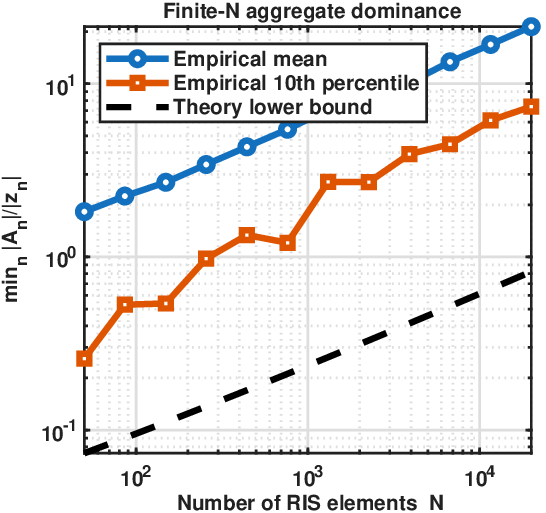}
    \caption{Empirical dominance ratio $\min_n|A_n|/|z_n|$ versus the number of RIS elements $N$ (log-log scale). The dashed line is the theoretical lower bound $(c_\delta/2)\sqrt{N}/\log(4N/\delta)$ from Theorem~\ref{thm:dominance}, with $c_\delta=\sqrt{\delta}/2$ and $\delta=0.1$.}
         
    \label{fig:dominance}
\end{figure}

% \begin{figure}[t]
%     \centering
%     \includegraphics[width=0.9\linewidth]{fig_discrete_curves.eps}
%     \caption{Achieved power ratio $|S|^2/(S^\ast)^2$ of the geometric
%     scheme versus the hardware phase resolution $b$, for feedback resolutions
%     $L\in\{2,4,6,8\}$ bits, at $N=1000$. Solid curves are the simulated
%     performance; the dashed curves are the worst-case lower bound $\cos^2\big(\min(\pi/2^{L}+\pi/2^{b},\pi/2)\big)$.}
%     \label{fig:discrete_curves}
% \end{figure}

% \begin{figure}[t]
%     \centering
%     \includegraphics[width=0.9\linewidth]{fig_discrete_heatmap.eps}
%     \caption{Simulated achieved power ratio over the feedback/hardware resolution
%     plane $(L,b)$, at $N=1000$. The dashed line is the equal-split locus $L=b$.}
%     \label{fig:discrete_heatmap}
% \end{figure}

\begin{figure}[t]
    \centering
    \includegraphics[width=0.95\linewidth]{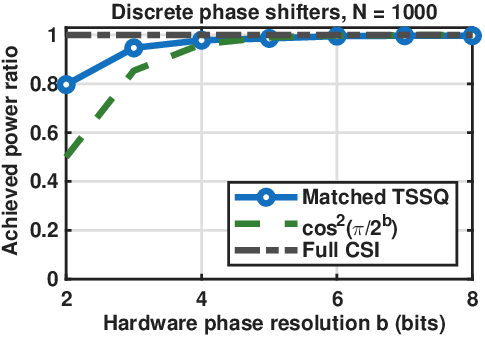}
    \caption{Achieved power ratio $\widehat\rho_{K,b}/\rho^\star_{\mathrm{RIS}}$
    of the geometric scheme with discrete phase shifters, versus the hardware resolution $b$ ($N=1000$). "Matched TSSQ" uses the codebook-matched TSSQ and the horizontal line at $1$ is the full-CSI value.}
    \label{fig:discrete_matched}
\end{figure}

% Fig.~\ref{fig:discrete_matched} compares the two ways of realizing a $b$-bit phase against the full-CSI value (the horizontal line at $1$). Matching the TSSQ leaves to the hardware alphabet is uniformly better than the cascade at every
% resolution. It reaches within a few percent of full CSI by $b=4$, a resolution the cascade attains only around $b=5$. This is precisely the one-bit
% advantage predicted by Corollary~\ref{cor:discrete}: the cascade's two roundings sum to a worst-case residual $2\pi/2^b$, whereas the single rounding of the matched scheme leaves only $\pi/2^b$. The empirical curves sit above their respective $\cos^2(\cdot)$ bounds, since those bounds are worst-case over the
% elements, while most elements are quantized far more accurately. In short, incorporating the hardware alphabet into the encoder provides the performance of the full-CSI system   %recovers close to the full-CSI beamforming gain 
% with only a handful of feedback bits per element. %, and without any nearest-codeword search.

Fig.~\ref{fig:discrete_matched} shows the power ratio achieved by the discrete scheme against the full-CSI value (the horizontal line at $1$). Confining every phase to the $b$-bit alphabet $\Phi_b$ costs remarkably little. The ratio climbs rapidly with $b$ and is already within a few percent of the full CSI by $b=4$. Therefore, a handful of feedback bits per element suffice to recover essentially the entire beamforming gain. The empirical curve sits above its worst case bound $\cos^2(\pi/2^b)$ of Corollary~\ref{cor:discrete}, since that bound is the worst case over the elements whereas most elements are quantized far more accurately. %Because the phase never leaves the grid, a single quantization of size $\Delta_b=\pi/2^b$ governs the loss, and the encoder attains this operating point with exactly $b$ bits per element. %and no codebook search.

% This behavior is summarized in Fig.~\ref{fig:discrete_heatmap}, which reports the simulated ratio over the entire $(L,b)$ plane. Performance is governed by the
% coarser of the two resolutions: the contours are roughly symmetric in $L$
% and $b$, and the map is dark (low power) whenever either axis is small. The
% equal-split locus $L=b$ (dashed) therefore traces an efficient operating curve.
% For a given total number of control bits, it is wasteful to make one resolution
% much finer than the other. Beyond about five bits on each axis, the achieved ratio
% exceeds $0.98$, confirming that near-optimal beamforming is attainable with modest, and jointly balanced, feedback and hardware resolutions.

\begin{figure}[t]
    \centering
    \begin{subfigure}{0.95\linewidth}
        \centering
        \includegraphics[width=\linewidth]{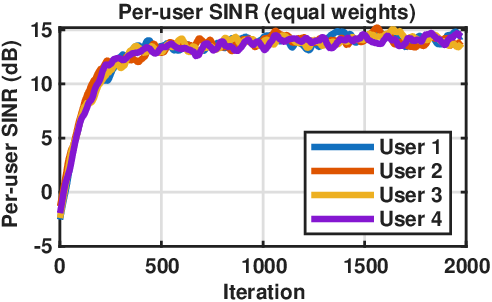}
        \caption{Equal weights, $\alpha=[1,1,1,1]$.}
        \label{fig:sinr_equal}
    \end{subfigure}

    \vspace{1ex}

    \begin{subfigure}{0.95\linewidth}
        \centering
        \includegraphics[width=\linewidth]{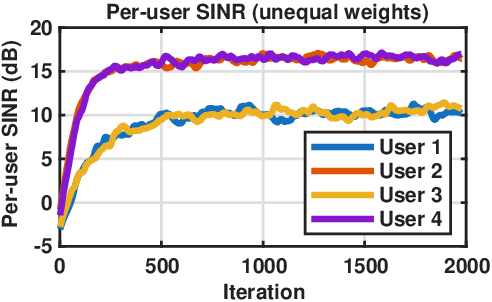}
        \caption{Unequal weights, $\alpha=[1,2,1,2]$.}
        \label{fig:sinr_unequal}
    \end{subfigure}
    \caption{Per-user SINR evolution under single-bit feedback for (a) equal and
    (b) unequal fairness weights ($K=4$, $N_{\rm T}=8$, $N=100$).}
    \label{fig:sinr}
\end{figure}

Fig.~\ref{fig:sinr} shows the per-user SINR trajectories in a multi-user scenario with $K=4$ users. Under equal weights (Fig.~\ref{fig:sinr_equal}), the four users converge to essentially the same SINR, so the system shares the RIS evenly among them. Under the unequal weights $\alpha=[1,2,1,2]$ (Fig.~\ref{fig:sinr_unequal}), the shared RIS configuration prioritizes Users 2 and 4, increasing their SINRs relative to the equal-weight case while reducing those of Users 1 and 3.
% Under the unequal weights $\alpha=[1,2,1,2]$ (Fig.~\ref{fig:sinr_unequal}), the shared RIS favors the more heavily weighted users, which stabilize at higher SINRs while the lightly weighted users are correspondingly lowered. 
Therefore, the weights control the operating point by adjusting the RIS phases and beamformer according to each user's requirements or priority.

Fig.~\ref{fig:selfish_consensus} compares the two beamforming updates (selfish and consensus methods) in terms of the achieved sum rate. The consensus update, in which every user reports a bit on each probed column, settles at a higher sum rate compared with the selfish rule. The gap reflects the interference that each column creates for the other users, and accounting for it, through the extra votes, lets the columns be steered cooperatively.

\begin{figure}[t]
    \centering
    \includegraphics[width=0.95\linewidth]{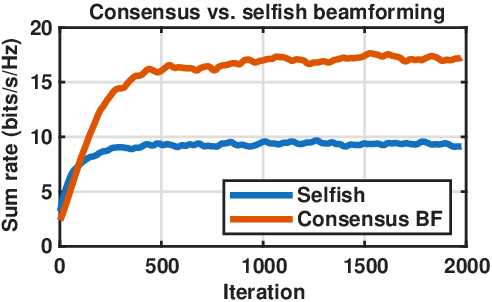}
    \caption{Sum rate of the consensus (weighted-majority) beamforming update
    versus the selfish per-user rule, averaged over channel realizations
    ($K=4$, $N_{\rm T}=8$, $N=100$).}
    \label{fig:selfish_consensus}
\end{figure}

\section{Conclusion}
\label{Conclusion}

We presented a channel-estimation-free framework for beamforming and RIS phase adaptation using only single-bit RSS feedback. The proposed methods avoid explicit cascaded-channel estimation and high-dimensional channel quantization while steering the transmit beamformer and RIS phases toward coherent combining. A key component is a deterministic geometric RIS update that reconstructs each phase correction from RSS-only probes. Under cascaded Rayleigh fading and channel-independent probing phases, we proved aggregate dominance and reference-direction stability with explicit $O(\log N/\sqrt{N})$ rates, established a global perturbation bound, and derived a received-power-ratio limit under vanishing quantization, direction, and normalized measurement errors. We also extended the framework to $b$-bit RIS phase shifters, yielding the closed-form gap $\cos^2(\pi/2^L+\pi/2^b)$, and to the multi-user downlink, where user-wise RSS feedback improves power balancing and spectral efficiency without multi-user cascaded-channel estimates. The proposed RSS-only adaptation approaches the full-CSI AO benchmark and outperforms channel-estimation-based limited-feedback baselines.

\appendices

\section{Proof of Lemma~\ref{lem:tail}}
\label{App:A}

Let $R_1=|X|$ and $R_2=|Y|$. Since $X,Y\sim\mathcal{CN}(0,1)$, the squared magnitudes $R_1^2$ and $R_2^2$ are independent exponential random variables with unit mean. Hence, for every $u\ge0$,
\begin{equation}
\Pr(R_i\ge u)=\Pr(R_i^2\ge u^2)=e^{-u^2},\qquad i=1,2.
\end{equation}
Now, fix $t>0$. If $R_1R_2\ge t$, then both $R_1$ and $R_2$ cannot be smaller than $\sqrt t$; otherwise, $R_1<\sqrt t$ and $R_2<\sqrt t$ would imply $R_1R_2<t$. Therefore,
\begin{equation}
\{R_1R_2\ge t\}
\subseteq
\{R_1\ge\sqrt t\}\cup\{R_2\ge\sqrt t\}.
\end{equation}
Using the union bound and the Rayleigh tail above,
\begin{align}
\Pr(|Z|\ge t)
&=\Pr(R_1R_2\ge t) \nonumber\\
&\le \Pr(R_1\ge\sqrt t)+\Pr(R_2\ge\sqrt t) \nonumber\\
&=2e^{-t}.
\end{align}
This proves the lemma.

\section{Proof of Lemma~\ref{lem:smallball}}
\label{App:B}
For $t\in\mathbb{C}$, define the two-dimensional characteristic function
\begin{equation}
\varphi_Z(t)=\mathbb{E}\left[e^{j\mathrm{Re}(\bar t Z)}\right].
\end{equation}
Conditioned on $X$, the product $Z=XY$ is complex Gaussian with distribution
$\mathcal{CN}(0,|X|^2)$. Write $t=t_{\mathrm r}+jt_{\mathrm i}$ and
$Z=Z_{\mathrm r}+jZ_{\mathrm i}$. Then, conditioned on $X$, the real and
imaginary parts $Z_{\mathrm r}$ and $Z_{\mathrm i}$ are independent
$\mathcal{N}(0,|X|^2/2)$ random variables, and
\begin{equation}
\mathrm{Re}(\bar t Z)=t_{\mathrm r}Z_{\mathrm r}+t_{\mathrm i}Z_{\mathrm i}.
\end{equation}
Using the characteristic function of a real Gaussian random variable,
\begin{align}
\mathbb{E}\left[e^{j\mathrm{Re}(\bar t Z)}\mid X\right]
&=
\exp\left(-\frac{|X|^2t_{\mathrm r}^2}{4}\right)
\exp\left(-\frac{|X|^2t_{\mathrm i}^2}{4}\right) \nonumber\\
&=
\exp\left(-\frac{|X|^2|t|^2}{4}\right).
\end{align}
Finally, since $|X|^2\sim\mathrm{Exp}(1)$,
\begin{align}
\varphi_Z(t)
&=
\int_0^\infty
\exp\left(-\frac{x|t|^2}{4}\right)e^{-x}\,\mathrm{d}x \nonumber\\
&=
\int_0^\infty
\exp\left[-x\left(1+\frac{|t|^2}{4}\right)\right]\,\mathrm{d}x \nonumber\\
&=
\frac{1}{1+|t|^2/4}.
\end{align}
Since $Z_m$ is circularly symmetric, multiplying it by the deterministic phase
$e^{j\phi_m}$ does not change its distribution. Hence,
$W_M\stackrel{d}{=}\sum_{m=1}^M Z_m$. Define the normalized sum
$U_M=W_M/\sqrt M$. By independence,
\begin{equation}
\varphi_{U_M}(t)
=
\left(\varphi_Z\left(\frac{t}{\sqrt M}\right)\right)^M
=
\left(1+\frac{|t|^2}{4M}\right)^{-M}.
\end{equation}
For $M\ge2$, this characteristic function is integrable over $\mathbb{R}^2$.
Therefore, by Fourier inversion, $U_M$ has a bounded continuous density
$f_{U_M}$, and
\begin{equation}
\|f_{U_M}\|_\infty
\le
\frac{1}{(2\pi)^2}
\int_{\mathbb{R}^2}
\left(1+\frac{|t|^2}{4M}\right)^{-M}\mathrm{d}t,
\end{equation}
where $\|f\|_\infty=\sup_{x\in\mathbb{R}^2}|f(x)|$.
Using polar coordinates, $t=\rho(\cos\psi,\sin\psi)$, we obtain
\begin{align}
\int_{\mathbb{R}^2}
\left(1+\frac{|t|^2}{4M}\right)^{-M}dt
&=
2\pi\int_0^\infty
\rho\left(1+\frac{\rho^2}{4M}\right)^{-M}d\rho .
\end{align}
Now, we set
\begin{equation}
s=1+\frac{\rho^2}{4M},
\qquad
ds=\frac{\rho}{2M}\,d\rho .
\end{equation}
Then
\begin{align}
2\pi\int_0^\infty
\rho\left(1+\frac{\rho^2}{4M}\right)^{-M}d\rho
&=
4\pi M\int_1^\infty s^{-M}ds \nonumber\\
&=
\frac{4\pi M}{M-1},
\end{align}
where the last equality holds for $M\ge2$. Therefore,
\begin{align}
\|f_{U_M}\|_\infty
&\le
\frac{1}{(2\pi)^2}\cdot\frac{4\pi M}{M-1} \nonumber\\
&=
\frac{M}{\pi(M-1)}
\le
\frac{2}{\pi},
\qquad M\ge2.
\end{align}
Finally, fix any $u\in\mathbb{C}$ and $r>0$. Since $U_M=W_M/\sqrt M$,
\begin{align}
\Pr(|W_M-u|\le r)
&=
\Pr\left(\left|U_M-\frac{u}{\sqrt M}\right|
\le
\frac{r}{\sqrt M}\right).
\end{align}
The event on the right is a disk in $\mathbb{R}^2$ with radius $r/\sqrt M$ and area
$\pi r^2/M$. Hence, using the density bound,
\begin{align}
\Pr(|W_M-u|\le r)
&\le
\|f_{U_M}\|_\infty
\frac{\pi r^2}{M} \nonumber\\
&\le
\frac{2}{\pi}\cdot\frac{\pi r^2}{M}
=
\frac{2r^2}{M}.
\end{align}
This proves the claim.

\section{Proof of Theorem~\ref{thm:dominance}}
\label{App:C}
For every $n$, the reverse triangle inequality applied to $A_n=S-z_ne^{j\theta_n}$ and $|z_ne^{j\theta_n}|=|z_n|$ gives $|A_n|\ge|S|-|z_n|\ge|S|-\max_k|z_k|$. Taking the minimum over $n$,
\begin{equation}
\label{eqn:dom-reduction}
\min_{1\le n\le N}|A_n|\ \ge\ |S|-\max_{1\le n\le N}|z_n|.
\end{equation}
 Applying Lemma~\ref{lem:smallball} with $M=N$, $u=0$, $r=c_\delta\sqrt N$, we get
\begin{equation}
    \Pr\bigl(|S|\le c_\delta\sqrt N\bigr)\le \frac{2(c_\delta\sqrt N)^2}{N}=2c_\delta^2=2\cdot\tfrac{\delta}{4}=\tfrac{\delta}{2},
\end{equation}
for $c_\delta=\tfrac12\sqrt\delta$. Hence, the event $\mathcal{E}_1=\{|S|\ge c_\delta\sqrt N\}$ has $\Pr(\mathcal{E}_1)\ge1-\delta/2$.
By Lemma~\ref{lem:tail}, $\Pr(|z_n|\ge t)\le2e^{-t}$ for each $n$. The union bound gives $\Pr(\max_n|z_n|\ge t)\le 2Ne^{-t}$. Setting $t=\log(4N/\delta)$ results in $2Ne^{-t}=2N\cdot\tfrac{\delta}{4N}=\tfrac{\delta}{2}$, so the event $\mathcal{E}_2=\{\max_n|z_n|\le\log(4N/\delta)\}$ has $\Pr(\mathcal{E}_2)\ge1-\delta/2$. By the union bound $\Pr(\mathcal{E}_1\cap\mathcal{E}_2)\ge1-\delta$. On $\mathcal{E}_1\cap\mathcal{E}_2$, \eqref{eqn:dom-reduction} gives $\min_n|A_n|\ge c_\delta\sqrt N-\log(4N/\delta)$ and the second inequality of \eqref{eqn:dom1} is exactly $\mathcal{E}_2$. If $c_\delta\sqrt N\ge2\log(4N/\delta)$, on the event $\mathcal{E}_1\cap\mathcal{E}_2$, we will have $\min_n|A_n|\ge c_\delta\sqrt N-\tfrac12 c_\delta\sqrt N=\tfrac{c_\delta}{2}\sqrt N$. Dividing by $\max_n|z_n|\le\log(4N/\delta)$ gives $\frac{\min_n |A_n|}{\max_n |z_n|}
    \ge
    \frac{c_\delta \sqrt{N}}
         {2\log\left(4N/\delta\right)}$.

\section{Proof of Theorem~\ref{thm:direction}}
\label{App:D}
Fix $n$ and consider the event $\mathcal{E}_1 \cap \mathcal{E}_2$ from Theorem~\ref{thm:dominance}, which occurs with probability at least $1-\delta$. Identify complex numbers with vectors in $\mathbb{R}^2$ and write $S = A_n + \boldsymbol{e}_n$, where $\boldsymbol{e}_n = z_n e^{j\theta_n}$. Then, we have $\|\boldsymbol{e}_n\| = |z_n|$ and $\|A_n\| = |A_n|$.
If $c_\delta \sqrt{N} \ge 3\log(4N/\delta)$, Theorem~\ref{thm:dominance} gives
\begin{align}
\frac{|z_n|}{|A_n|}
\le
\frac{\log(4N/\delta)}
     {c_\delta \sqrt{N} - \log(4N/\delta)}
&\le
\frac{\log(4N/\delta)}
     {3\log(4N/\delta) - \log(4N/\delta)} \nonumber\\
&=
\frac{1}{2}.
\end{align}
Therefore, $\|\boldsymbol{e}_n\| < \|A_n\|$ and applying Lemma~\ref{lem:angle} with $\boldsymbol{q}=A_n$ results in
\begin{equation}
|\angle S - \angle A_n|
\le
\arcsin\left(\frac{|z_n|}{|A_n|}\right)
\le
\frac{\pi}{2}\frac{|z_n|}{|A_n|}.
\end{equation}

Moreover, $c_\delta\sqrt N\ge3\log(4N/\delta)$ gives $|A_n|\ge c_\delta\sqrt N-\log(4N/\delta)\ge\tfrac{2}{3}c_\delta\sqrt N$, hence $|A_n|\ge\tfrac{c_\delta}{2}\sqrt N$, and with $|z_n|\le\log(4N/\delta)$,
\[
|\angle S-\angle A_n|\le\frac{\pi}{2}\cdot\frac{\log(4N/\delta)}{(c_\delta/2)\sqrt N}
=\frac{\pi}{c_\delta}\frac{\log(4N/\delta)}{\sqrt N}
\le\frac{4}{c_\delta}\frac{\log(4N/\delta)}{\sqrt N}.
\]
Taking the maximum over $n$ (the bound is uniform in $n$ on the same event) proves \eqref{eqn:dir}.

\section{Proof of Theorem~\ref{thm:lipschitz}}
\label{App:E}
Define $\boldsymbol{q}(d_1,d_2)=(-d_1,\,d_1-2d_2)\in\mathbb{R}^2$. Since $\operatorname{atan2}(y,x)$ is the angle of the vector $(x,y)$, we have $F(d_1,d_2)=\angle\boldsymbol{q}(d_1,d_2)$, and $\|\boldsymbol{q}^\star\|=\sqrt{(d_1^\star)^2+(d_1^\star-2d_2^\star)^2}=R^\star$. (The identity $R^\star=4|A_n||z_n|$ holds because, with $a=|A_n|$, $b=|z_n|$ and $\alpha$ the angle of element $n$ relative to $A_n$, the law of cosines gives $d_1^\star=-4ab\cos\alpha$ and $d_1^\star-2d_2^\star=4ab\sin\alpha$, so $\boldsymbol{q}^\star=4ab(\cos\alpha,\sin\alpha)$, where $\|\boldsymbol{q}^\star\|=4ab$ and $\angle\boldsymbol{q}^\star=\alpha=\alpha^\star$.) The measurement perturbation is $\widehat{\boldsymbol{q}}-\boldsymbol{q}^\star=(-\Delta_1,\ \Delta_1-2\Delta_2)$, and by $\sqrt{p^2+q^2}\le|p|+|q|$,
\begin{align}
\label{eqn:e-bound}
\|\widehat{\boldsymbol{q}}-\boldsymbol{q}^\star\|=\sqrt{\Delta_1^2+(\Delta_1-2\Delta_2)^2}&\le|\Delta_1|+|\Delta_1-2\Delta_2|\nonumber \\ &\le2(|\Delta_1|+|\Delta_2|).
\end{align}
The assumption $|\Delta_1|+|\Delta_2|\le R^\star/4$ then gives $\|\widehat{\boldsymbol{q}}-\boldsymbol{q}^\star\|\le R^\star/2<\|\boldsymbol{q}^\star\|$. Therefore, Lemma~\ref{lem:angle} (with $\boldsymbol{q}=\boldsymbol{q}^\star$, $\boldsymbol{e}=\widehat{\boldsymbol{q}}-\boldsymbol{q}^\star$) results in
\begin{align}
\label{eqn:est-arcsin}
|\widehat\alpha-\alpha^\star|&\le\arcsin\Bigl(\frac{\|\widehat{\boldsymbol{q}}-\boldsymbol{q}^\star\|}{R^\star}\Bigr)\le\frac{\pi}{2}\frac{\|\widehat{\boldsymbol{q}}-\boldsymbol{q}^\star\|}{R^\star}\nonumber \\ &\le\frac{\pi(|\Delta_1|+|\Delta_2|)}{R^\star}\le\frac{4(|\Delta_1|+|\Delta_2|)}{R^\star},
\end{align}
which is the same as \eqref{eqn:lip}.

To prove the second part, fix any $x\in(0,\pi]$ and define
$D=|\Delta_1|+|\Delta_2|$ and $H=\{D\le R^\star/4\}$. Then,
\[
\{|\widehat\alpha-\alpha^\star|>x\}
=
\left(\{|\widehat\alpha-\alpha^\star|>x\}\cap H\right)
\cup
\left(\{|\widehat\alpha-\alpha^\star|>x\}\cap H^c\right),
\]
so we can bound the two cases separately. From \eqref{eqn:lip},
\begin{equation}
\begin{aligned}
\left\{|\widehat\alpha-\alpha^\star|>x\right\}
\subseteq
&
\underbrace{
\left\{
D>\frac{xR^\star}{4}
\right\}
}_{\text{small-perturbation region}}
\cup
\underbrace{
\left\{
D>\frac{R^\star}{4}
\right\}
}_{\text{escape event}} .
\end{aligned}
\end{equation}
On $H$, \eqref{eqn:lip} gives $|\widehat\alpha-\alpha^\star|\le 4D/R^\star$, so $|\widehat\alpha-\alpha^\star|>x$ implies $D>xR^\star/4$. On $H^c$, we directly have $D>R^\star/4$, which is the escape event.
Using $\{|a|+|b|>u\}\subseteq\{|a|>u/2\}\cup\{|b|>u/2\}$ and the tail bound $\Pr(|\Delta_i|>u)\le 2\exp(-u^2/(2\sigma^2))$, we obtain
\begin{equation}
\begin{aligned}
\Pr\left(
|\Delta_1|+|\Delta_2|>\frac{xR^\star}{4}
\right)
&\le
4\exp\left(
-\frac{(xR^\star/8)^2}{2\sigma^2}
\right)
\\
&=
4\exp\left(
-\frac{x^2(R^\star)^2}{128\sigma^2}
\right).
\end{aligned}
\end{equation}
Similarly,
\begin{equation}
\Pr\left(
|\Delta_1|+|\Delta_2|>\frac{R^\star}{4}
\right)
\le
4\exp\left(
-\frac{(R^\star)^2}{128\sigma^2}
\right).
\end{equation}
Adding the two bounds gives
\begin{equation}
\Pr\left(|\widehat\alpha-\alpha^\star|>x\right)
\le
4\exp\left(
-\frac{x^2(R^\star)^2}{128\sigma^2}
\right)
+
4\exp\left(
-\frac{(R^\star)^2}{128\sigma^2}
\right),
\end{equation}
which proves the claim.

\bibliography{journalAbbreviations, mybibliography2}
\bibliographystyle{ieeetr}
\end{document}